\documentclass[a4paper,11pt]{article}

\usepackage[english]{babel}
\usepackage[latin1]{inputenc}        
\usepackage[T1]{fontenc}
\usepackage{csquotes}
\usepackage{pdfsync}

\usepackage{amsmath,amsfonts,amsthm,amsbsy}
\usepackage[svgnames,x11names]{xcolor}

\numberwithin{equation}{section}

\usepackage{geometry}
\newgeometry{margin=2.5cm}

\usepackage{tikz}
\usetikzlibrary{arrows}
\usetikzlibrary{intersections}

\usepackage{graphicx}
\usepackage{caption}
\usepackage{subcaption}

\definecolor{ct_black}{HTML}{000000}
\definecolor{ct_orange}{HTML}{ED872D}
\definecolor{ct_purple}{HTML}{7A68A6}
\definecolor{ct_blue}{HTML}{348ABD}
\definecolor{ct_turquoise}{HTML}{188487}
\definecolor{ct_red}{HTML}{E32636}
\definecolor{ct_pink}{HTML}{CF4457}
\definecolor{ct_green}{HTML}{467821}

\definecolor{ct2_green}{HTML}{9FF781}
\definecolor{ct2_green_dark}{HTML}{088A08}

\newcommand{\norm}[1]{\left\lVert#1\right\rVert}

\theoremstyle{plain}
\newtheorem{thm}{Theorem}[section]
\newtheorem{cor}[thm]{Corollary}
\newtheorem{lem}[thm]{Lemma}
\newtheorem{prop}[thm]{Proposition}

\theoremstyle{definition}
\newtheorem{defn}[thm]{Definition}
\newtheorem{rem}[thm]{Remark}

\newcommand{\ee}{{\mathrm e}}
\newcommand{\ii}{{\mathrm i}}
\newcommand{\dd}{{\mathrm d}}

\newcommand{\tr}{{\mathrm{tr} }}
\newcommand{\Tr}{{\mathrm{Tr}}}

\newcommand{\tand}{{\qquad \mathrm{and} \qquad}}

\newcommand{\HB}{H_\mathrm{B}}
\newcommand{\HBo}{H_\mathrm{B,1}}
\newcommand{\HBt}{H_\mathrm{B,2}}
\newcommand{\HBz}{H_\mathrm{B,0}}
\newcommand{\HBs}{H_{\mathrm{B},s}}
\newcommand{\UB}{U_\mathrm{B}}
\newcommand{\UBo}{U_\mathrm{B,1}}
\newcommand{\UBt}{U_\mathrm{B,2}}
\newcommand{\HHB}{\mathcal{H}_\mathrm{B}}

\newcommand{\HE}{H_\mathrm{E}}
\newcommand{\UE}{U_\mathrm{E}}
\newcommand{\HHE}{\mathcal{H}_\mathrm{E}}

\newcommand{\HI}{H_\mathrm{I}}
\newcommand{\UI}{U_\mathrm{I}}
\newcommand{\DI}{D_\mathrm{I}}
\newcommand{\DIo}{D_{\mathrm{I},1}}
\newcommand{\DIt}{D_{\mathrm{I},2}}

\newcommand{\Hbc}{H_\mathrm{bc}}

\newcommand{\IB}{\mathcal{I}_\mathrm{B}}
\newcommand{\IE}{\mathcal{I}_\mathrm{E}}
\newcommand{\II}{\mathcal{I}_\mathrm{I}}

\newcommand{\vm}{\mathbf m}
\newcommand{\vn}{\mathbf n}
\newcommand{\vp}{\mathbf p}

\newcommand{\Id}{\mathrm{Id}}
\newcommand{\idB}{I}
\newcommand{\idE}{I}

\usepackage[pdftex,%
colorlinks=true,linkcolor=blue,citecolor=red,%
]
{hyperref}

\begin{document}

\title{Bulk-Edge Correspondence for Two-Dimensional \\ Floquet Topological Insulators}

\author{Gian Michele Graf  and Cl\'ement Tauber
 \\	\footnotesize{Institute for Theoretical Physics, ETH Z\"{u}rich }} 
\date{}

\maketitle

\begin{abstract}
Floquet topological insulators describe independent electrons on a lattice driven out of equilibrium by a time-periodic Hamiltonian, beyond the usual adiabatic approximation. In dimension two such systems are characterized by integer-valued topological indices associated to the unitary propagator, alternatively in the bulk or at the edge of a sample. In this paper we give new definitions of the two indices, relying neither on translation invariance nor on averaging, and show that they are equal. In particular weak disorder and defects are intrinsically taken into account. Finally indices can be defined when two driven sample are placed next to one another either in space or in time, and then shown to be equal. The edge index is interpreted as a quantized pumping occurring at the interface with an effective vacuum. 

\end{abstract}

\section{Introduction}

Bulk-edge correspondence is a crucial concept in the context of Quantum Hall effect and topological insulators. From the topological point of view, the bulk properties of an infinite sample can be deduced by looking at the gapless modes, propagating at the edge of a sample with boundary, and \textit{vice versa} \cite{Hatsugai93,ElgartGrafSchenker05,GrafPorta13}. This duality is commonly observed in physical systems where both bulk and edge index are well understood. Sometimes it is even assumed to fill the lack of interpretation of a bulk invariant, the physics at the edge being usually more intuitive. In any case a proof of this correspondence is as much a mathematical challenge as a helpful identity for physics. 

In analogy with topological insulators, it was recently realized that topological phases could arise in periodically driven systems. The initial proposal was to induce topology on a two-dimensional sample through a time-periodic perturbation of a trivial material, e.g. by irradiation of graphene \cite{OkaAoki09,InoueTanaka10} or semi-conductor quantum wells by microwaves \cite{LindnerRefaelGalitski11}, but it was then realized that a large class of time-periodic Hamiltonians of independent electrons may support topological properties, as long as the unitary propagator after one period is gapped  \cite{KitgawaBergRudnerDemler10,RudnerPRX13}.

For samples that are also space-translation invariant, Rudner \textit{et al.} \cite{RudnerPRX13} defined a topological bulk index that is integer-valued and equal to the number of edge modes that appear in the spectrum for associated dynamics on a strip geometry. Moreover an explicit definition of the edge index and a proof of the bulk-edge correspondence was proposed in \cite{RudnerPRX13}, but with the extra assumption that the unitary propagator is also periodic in time. Recently the requirement of spatial invariance has been dropped and similar result were obtained for disordered systems \cite{FulgaMaksymenko16}, or \cite{TitumPRX16} where averaging over fluxes threading the sample has been used. An interacting model was proposed in \cite{KlinovajaStanoLoss16}. Finally the bulk invariant has been generalized to the cases with time-reversal or chiral symmetry \cite{Lyon15bis, Lyon15, Fruchart16}, and bulk-edge correspondence for one-dimensional chiral systems was studied in \cite{AsbothTarasinskiDelplace14}.

In this paper we give new definitions both for the bulk and edge index that do not require space-translation invariance of the Hamiltonian, nor averaging, and show a general proof of the bulk-edge correspondence. We only assume that the Hamiltonian is local (short range), periodic  and regular enough in time. The construction works as soon as the bulk one-period propagator  has a spectral gap. If space-translation invariance is present though, the definition generalizes the existing one. If not, it applies to weakly disordered systems, see Rem.\,\ref{rem:disorder}. Moreover in this approach the edge index is interpreted as a quantized pumping of charges after one cycle. Exploiting a duality between space and time (see Sect.\,\ref{subsec:interface}), we show that this pumping actually occurs at the interface with an effective vacuum, computed from the original Hamiltonian and depending on the spectral gap under consideration. 

The concept of topological pump and the study of periodically driven system in this context is not new but until recently the adiabatic hypothesis has been always implied. From Thouless' original work \cite{Thoules83} to more recent and abstract considerations \cite{ProdanSchulz16book}, the driving was always assumed to be slow enough in order to use the adiabatic theorem. In particular the time-dependent spectrum of the Hamiltonian is the relevant object of interest, and usually a persistent gap all along the driving is assumed. We stress that Floquet topological insulators and in particular the present work are not placed in this frame. Here the driving can be arbitrary and we do not make any assumption on the spectrum of the Hamiltonian, but only on the corresponding propagator. Finally note that this notion of non-adiabatic quantized pumping has already been observed in \cite{TitumPRX16}.

The paper is organized as follows. First Sect.\,\ref{sec:FTI} describes the context of Floquet topological insulators for which the construction applies. The main results are then stated in Sect.\,\ref{sec:main_results}. The definition of bulk and edge indices, as well as the bulk-edge correspondence, is done in two steps. Inspired by \cite{RudnerPRX13}, we first assume that the bulk propagator is periodic in time. The edge invariant is interpreted as charge pumping and can be identified with an index of pair of projections~\cite{AvronSeilerSimon94}. The bulk index is a mixture of commutative (in time) and non-commutative (in space) expression of the odd Chern number \cite{ProdanSchulz16}. For the general case we define the bulk and edge index through a relative time evolution that allows to reduce matters to the previous case, by considering an effective Hamiltonian for each spectral gap of the bulk propagator. The index of an interface is also defined to provide a simple interpretation of this effective Hamiltonian.

Sect.\,\ref{sec:timeevol} then studies the locality and continuity properties of bulk and edge propagators, required for the indices to be well-defined, and compare these propagators. All this is established through the notion of confinement \cite{ElgartGrafSchenker05} and switch functions \cite{AvronSeilerSimon94}. The proofs are finally detailed in Sect.\,\ref{sec:proofs}, mostly following the statements of Sect.\,\ref{sec:main_results} but postponing some computations to App.\,\ref{sec:app}. Although the mathematical expressions of the indices look similar to those for topological insulators, the operators involved are quite different and indeed describe another physics. 

Finally note that shortly after this work was completed an independent result on similar matters was proposed in \cite{SadelSchulz17}. Based on K-theory, it extends this bulk-edge correspondence to every dimension, but the physical interpretation is less immediate than in the functional analysis approach. Moreover our work does not rely on any covariance property.

\section{Floquet topological insulators \label{sec:FTI}}

\subsection{Bulk and edge Hamiltonians}

We consider a tight-binding model of independent electrons on the two-dimensional lattice $\mathbb Z^2$. The bulk Hilbert space is $\HHB = \ell^2(\mathbb Z^2) \otimes \mathbb C^N$, where $\mathbb C^N$ accounts for internal degrees of freedom (sub-lattice, spin, orbital, etc.). For $\vm \in \mathbb Z^2$, we denote by the usual ket notation $|\vm \rangle \in \ell^2(\mathbb Z^2)$ the state localized at site $\vm$ and $\langle \vm |$ its corresponding bra. For any operator $K$ on $\HHB$ and $\vm,\vn \in \mathbb Z^2$, the kernel $K_{\vm,\vn} \equiv \langle \vm | K | \vn \rangle$ is a matrix of size $N$. According to the context $|\vm| = |m_1| + |m_2|$ and $|K_{\vm,\vn}|$ denotes the operator norm of finite matrices. The operator norm on the full Hilbert space $\HHB$ is denoted by $\norm{K}$. 

The electrons are ruled by a family of one-particle Hamiltonians $H_B(t)$, namely a self-adjoint operator on $\HHB$ for each $t \in \mathbb R$. In the context of Floquet topological insulators we assume that it satisfies some further assumptions. 

\begin{defn}[Bulk Hamiltonian]\label{def:bulk_Hamiltonian}
 Let $\HB(t) : \HHB \rightarrow \HHB$ be a family of self-adjoint operators for $t \in \mathbb R$. We say that $\HB$ is a bulk Hamiltonian if it is
 \begin{enumerate} 
 	\item \textit{time-periodic:} $\exists T \in \mathbb R$ so that $\HB(t+T) = \HB(t)$ for all $t \in \mathbb R$,
 	\item \textit{local:}  $ \exists \, \mu, \,C >0$ independent of $t$ so that for any $t \in [0,T]$ and  $\vm,\vn \in \mathbb Z^2$
 	\begin{equation}\label{HB_loc}
 	|\HB(t)_{\vm,\vn}| \leq C \ee^{-\mu |\vm-\vn|}\, ;
 	\end{equation}
 	$\mu$ is called the \textit{locality exponent},
  	\item \textit{piecewise strongly continuous:} the map $t \mapsto \HB(t)$ is strongly continuous except possibly for jump discontinuities.
 \end{enumerate}
\end{defn}

Note that because of Condition 1, the parameter $t$ is reduced to a compact interval  so that the uniform bound in Condition 2 is equivalent to a family of time-dependent bounds for $t \in [0,T]$.

\begin{rem}[Physical models covered]
 Any time-periodic Hamiltonian that for each $t$ is a finite range or exponentially decaying hopping term is a bulk Hamiltonian in the sense of Def.\,\ref{def:bulk_Hamiltonian}. Moreover piecewise constant Hamiltonians (e.g. as in \cite{RudnerPRX13}) are also allowed thanks to Condition~3. However we do not require space translation invariance for a bulk Hamiltonian so that any disordered configuration can be implemented through $\HB$ \textit{a priori}, see Rem.\,\ref{rem:disorder} below. Finally  we do not require  a spectral gap uniform in times, in contrast to adiabatic theory.
\end{rem}

\begin{rem}[Underlying topology]\label{rem:topology}  
	We define a norm on local operators which is suited to bulk Hamiltonians. For fixed $\mu$ let
	\begin{equation}\label{local_norm}
	\norm{A}_\mu = \inf \lbrace C \, | \,   \forall t \in [0,T] \quad \forall \vm,\vn \in \mathbb Z^2 \quad	|A(t)_{\vm,\vn}| \leq C \ee^{-\mu |\vm-\vn|} \rbrace,
	\end{equation}
	which satisfies $\norm{\cdot}_\lambda \leq \norm{\cdot}_\mu$ for $\lambda \leq \mu$.
	 %In particular bulk Hamiltonians are bounded uniformly in time.
	This local norm will be used for homotopy considerations.
\end{rem}

The edge system is described by considering only a half-plane, which we take to be $\mathbb  N \times \mathbb Z \subset \mathbb Z^2$, so that the edge Hilbert space is $\HHE = \ell^2(\mathbb N \times \mathbb Z) \otimes \mathbb C^N$. Bulk and edge spaces are related through the partial isometry
\begin{equation}
\iota : \HHE \longrightarrow \HHB, \qquad \iota^* : \HHB \longrightarrow \HHE, 
\end{equation}
where $\iota$ is the canonical injection of $\HHE$ in $\HHB$ and $\iota^*$ is the canonical truncation of $\HHB$ to $\HHE$. In particular they satisfy
\begin{equation}\label{J_partialisom}
\iota^* \iota = \Id_{\HHE}, \qquad \iota \iota^* = P_1,
\end{equation}
where $P_1 : \HHB \rightarrow \HHB$ is the projection on states supported in the right half-plane $n_1 \geq 0$.

\begin{defn}[Edge Hamiltonian]\label{def:edge_Hamiltonian}
	For a given bulk Hamiltonian $\HB(t)$, the edge Hamiltonian $\HE(t) : \HHE \rightarrow \HHE$ is the family of self-adjoint operators defined by
	\begin{equation}
	\HE(t) = \iota^* \HB(t) \iota.
	\end{equation}
	Properties 1-3 of $\HHB$, $\HB$ are inherited to $\HHE$, $\HE$. In particular $\norm{\HE}_\mu \leq \norm{\HB}_\mu$.
\end{defn}

As a sharp cut of the bulk space, this edge Hamiltonian corresponds to Dirichlet boundary condition, but an extra term confined near the boundary can actually be added to the previous definition without changing the topological aspects, see Prop.\,\ref{prop:general_BC} below, allowing the implementation of other local boundary conditions or defects at the edge.
 
\subsection{Propagator}

The spectrum of a time-dependent Hamiltonian $H(t)$ at any given time will not be of importance. Instead we shall consider the time evolution operator generated by $H(t)$, see e.g. \cite[Thm. X.69]{ReedSimonII}.

\begin{defn}[Propagator]
	Let $H(t)$ be a family of bounded Hamiltonians on a Hilbert space $\mathcal H$, with $t \mapsto H(t)$ strongly continuous. The unitary propagator $U(t,s) \in \mathcal U(\mathcal H)$ is a two parameter family of unitary operators strongly continuous in $t$ and $s$ satisfying 
	\begin{equation}\label{prop_propag}
	U(t,t) = \Id_\mathcal{H}, \qquad U(t,r)U(r,s) = U(t,s),
	\end{equation}
	and so that for any $\psi \in \mathcal H$, $\varphi_s(t) = U(t,s) \psi$ is the unique solution of
	\begin{equation}\label{Schrodinger_equation}
	\ii \dfrac{\dd\,}{\dd t} \varphi_s(t) = H(t) \varphi_s(t), \qquad \varphi_s(s) = \psi.
	\end{equation}
	where we have set $\hbar=1$. $H(t)$ is called the generator of $U(t,s)$.
\end{defn}
If the Hamiltonian has jump discontinuities, the propagator is defined piecewise but remains strongly continuous even at the discontinuity points thanks to \eqref{prop_propag}. Note that in the case of a time-independent Hamiltonian $H$, the propagator is given by
\begin{equation}
U(t,s) = \ee^{-\ii (t-s) H}
\end{equation}
and satisfies $U(t+\tau,s+\tau) = U(t,s)$ for any $\tau \in \mathbb R$. 
If $H(t+T)=H(t)$ is periodic in time, then that property survives for $\tau=T$, which implies $U(t+T,s) = U(t,0)U(T,s)$ by \eqref{prop_propag}. As a result the whole family $U(t,s)$ is determined by its restriction  $U(t) \equiv U(t,0)$ to the compact interval $0 \leq t \leq T$; and its long time behavior by just $U(T)$.
The spectrum of $U(T)$ thus carries  essential  information about the solutions of \eqref{Schrodinger_equation}. This is the so-called Floquet theory. Because $U(t)$ is unitary, its spectrum $\sigma[U(t)]$ belongs to $\mathcal S^1$, and at $t=T$ we denote
\begin{equation}
\ee^{-\ii \varepsilon T} \in \sigma[U(T)] \subset \mathcal S^1
\end{equation}
so that $\varepsilon$ has the dimension of an energy. Because it is defined modulo $2\pi/T$, it is rather called \textit{quasi-energy}, in analogy with quasi-momentum in Bloch theorem. Indeed the eigenstates of $U(T)$ provide solutions to \eqref{Schrodinger_equation} that are time-periodic up to the phase $\ee^{-\ii \varepsilon T}$.

\subsection{Stroboscopic gap assumption}

The topological aspects can be characterized through the propagator of a bulk Hamiltonian.

\begin{defn}[Floquet Topological Insulator]\label{def:FTI}
We say that $\HB(t)$, a bulk Hamiltonian in the sense of Def.\,\ref{def:bulk_Hamiltonian}, is a Floquet topological insulator if the corresponding unitary propagator at $t=T$, $\UB(T)$ has a spectral gap.
\end{defn}

\begin{figure}[htb]
	\centering
	\begin{tikzpicture}[scale=0.9]
	\newdimen\r
	\pgfmathsetlength\r{1.5cm}
	\draw[thick] (0,0) circle (\r);
	\draw[line width=0.15cm,DeepSkyBlue3] (20:\r) arc (20:130:\r);
	\draw[line width=0.15cm,DeepSkyBlue3] (-40:\r) arc 
	(-40:-120:\r);
	
	\draw (195:2.5*\r) node[text centered, text width=3cm]{quasi-energy \\ gaps};
	\draw (0:2.5*\r) node[text centered, text width=3cm]{quasi-energy \\ bands};
	
	\draw[-latex,dotted] (0:1.7*\r) -- (30:1.1*\r); 
	\draw[-latex,dotted] (0:1.7*\r) -- (-50:1.1*\r); 
	
	\draw[-latex,dotted] (200:1.7*\r) -- (-15:0.9*\r); 
	\draw[-latex,dotted] (200:1.7*\r) -- (170:1.1*\r); 
	
	\draw[dashed] (0,0) -- (150:1.2*\r) node[above left]{$\ee^{-\ii T\varepsilon}$};
	\end{tikzpicture}
	\caption{\label{fig:spectral_gap} Example of spectrum for $\UB(T)$ with two quasi-energy bands and gaps.}
\end{figure}
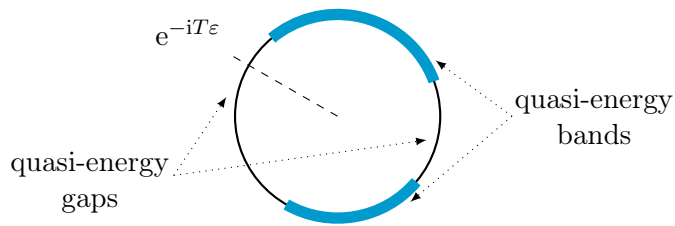

As illustrated in Fig.\,\ref{fig:spectral_gap}, the spectrum of $\UB(T)$ is typically constituted of one or several bands (of arbitrary nature) separated by gaps. By extension we also speak about quasi-energy $\varepsilon$ when $\ee^{-\ii T \varepsilon}$ is in a spectral gap of $\UB(T)$. Moreover note that assumption of a \enquote{stroboscopic} spectral gap, \textit{i.e.} only for $U_B(T)$, is sufficient to define the topological indices. The gap assumption may fail at some intermediate times, \textit{i.e.} for $\UB(t),\,0<t<T$.

\begin{rem}
	The term \enquote{insulator} is somewhat misleading here as its meaning is purely mathematical: The existence of a spectral gap. The physical interpretation is not obvious since the spectrum of a unitary operator lives on a circle, so there is no ground state (in fact energy is not even conserved) and thus no notion of Fermi energy. The analogy with (time-independent) topological insulators should then be used with care. Some attempt of interpretation is given in Sect.\,\ref{subsec:interface} below.
\end{rem}

\section{Bulk-edge correspondence \label{sec:main_results}}

The main result of this paper is to define a bulk and edge index and to show that they coincide, for each spectral gap of $\UB(T)$. The indices are respectively defined in terms of the bulk and edge propagators $\UB$ and $\UE$, generated by the corresponding Hamiltonians. To do so, the first thing to establish is that $\UB$ and $\UE$ are both local when $\HB$ is, see Sect.\,\ref{sec:localU} below. The operations of truncating space and generating time evolution do not commute, so that the truncated bulk propagator does not equal that of the edge. The important point however is that
\begin{equation}\label{pres_D}
|D(t)_{\vm,\vn}| \leq C \ee^{- \lambda |m_2-n_2|} \ee^{-\lambda|n_1|},  \qquad D(t) \equiv U_E(t) - \iota^* \UB(t) \iota
\end{equation}
for some $C>0$ and $0 < \lambda <\mu$, see Prop.\,\ref{prop:defD}. Namely the difference $D$ is confined near the edge since it is exponentially decaying in direction 1, compare with \eqref{HB_loc}. The bulk and edge indices are then defined using switch functions \cite{AvronSeilerSimon94}.
  
\begin{defn}\label{def:switch}
	A switch function $\Lambda : \mathbb Z \rightarrow \mathbb R$ is a function so that $\Lambda(n) = 1$ (resp. $0$) for $n$ large and positive (resp. negative).
	We also call switch function and denote by $\Lambda$ the multiplicative operator acting on $\ell^{2}(\mathbb Z)$, and by $\Lambda_i$ a switch function $\Lambda_i(\vn)=\Lambda(n_i)$ in direction $i$ acting on $\ell^2(\mathbb Z^2)$ or $\ell^2(\mathbb N \times \mathbb Z)$.
\end{defn}

For instance $\Lambda$ can be a step function, in which case it is a projection, such as $P_1$ in \eqref{J_partialisom}. The commutator with a switch function allows to confine a local operator in a particular direction and is a powerful tool to eventually end up with trace-class expressions. This is detailed in Sect.\,\ref{sec:local_conf_trace}.

\subsection{The case of a time-periodic propagator}

The definition and properties of the indices, as well as the bulk-edge correspondence, are first established under the auxiliary assumption that the bulk propagator satisfies:
\begin{equation}\label{UB=I}
\UB(T) =  \idB,
\end{equation}
where $\idB$ is the identity. Although not really physical, this situation still belongs to the Floquet Topological Insulators in the sense of Def.\,\ref{def:FTI} since the spectrum of $\UB(T)$ is degenerated to $\{1\}$ so that $\mathcal S^1 \setminus \{1\}$ constitutes a canonical spectral gap (see Fig.\,\ref{fig:IE} right). The general case, treated in the next section, is nothing but a reduction to this particular one. 

\begin{prop}[Edge index definition]\label{def:IEperiodic}
	Let $\HB$ be a bulk Hamiltonian so that $\UB(T) =  \idB$. Let $\HE$ and $\UE$ the associated edge Hamiltonian and propagator, and $\Lambda_2$ a switch function in direction 2 on $\HHE$. The edge index
	\begin{equation}\label{IE_def}
	\IE \equiv \Tr_{\HHE} \Big( U_E^*(T) [\Lambda_2, U_E(T)] \Big)
	\end{equation}
	is well-defined and integer valued, independent of the choice of $\Lambda_2$, and continuous in $\HB$ (in local norm)  as long as $\UB(T)=\idB$.
\end{prop}

In that case $\UE(T) = \idE + D(T)$ so that $\UE$ is time-periodic up to a correction confined at the edge. The index has the interpretation of a \textit{non-adiabatic quantized charge pumping} \cite{TitumPRX16}: It counts the net number of particles that have moved into the upper half-plane within a period. In fact, by the independence on $\Lambda_2$, we may pick $\Lambda_2 = P_2$, the projection associated to that half-plane, so that
 \begin{equation}\label{IEasindex}
  \IE = \Tr_{\HHE} \Big( U_E^*(T) P_2 U_E(T) - P_2 \Big) 
  \end{equation}
indeed computes the difference in the number of particles therein at times separated by a period. Moreover the net transport takes place near the edge because far away from it we may pretend $\UE(T) = 1$ by (\ref{pres_D}, \ref{UB=I}). See Fig.\,\ref{fig:IE}. As we shall see \eqref{IEasindex} is the index of a pair of projection \cite{AvronSeilerSimon94} and hence an integer.

 \begin{figure}[htb]
  	\centering
  	\begin{tikzpicture}
  	\filldraw[color=ct2_green,opacity=0.3 ] (0,0) -- (0,2.5)--(5,2.5)--(5,0) --cycle;
  	\draw[-latex,thick] (-1,0) -- (5,0) node[below]{$n_1 \in \mathbb N$};
  	\draw[-latex,thick] (0,-2.5) -- (0,2.5) node[above]{$n_2 \in \mathbb Z$};
  	\foreach \k in {-9,...,9}
  	{\draw (0,-2.5/10*\k) -- (-0.2,-2.5/10*\k-0.3);}
  	\draw[ct2_green_dark] (5,2.5) node[below left] {$n_2 \geq 0$};
  	\draw[red, very thick, ->] (0.2,-1.5)node {$\bullet$} -- (0.2,1.5);
  	\draw[->] (2.5-0.5,-0.5) arc (-90:-441:0.5); \draw (2.5-0.5,-0.5) node {$\bullet$};
  	\draw[->] (2.5+2,-0.5+1) arc (-90:-441:0.5); \draw (2.5+2,-0.5+1) node {$\bullet$};
  	\draw[->] (2.5+1.5,-0.5-1.5) arc (-90:-441:0.5); \draw (2.5+1.5,-0.5-1.5) node {$\bullet$}
  	;	
  	\begin{scope}[xshift=10.5cm]
  	\draw (0,0) circle (1.6);
  	\draw[dashed,thick,color=red] (0,0) circle (1.65);

  	\draw (0,-2.2) node{$\sigma\big[U_E(T)\big]$};
  	\draw[DeepSkyBlue3] (1.6,0) node {$\bullet$};
  	
  	\draw[-latex] (2.7,1.1) node[above]{$\sigma\big[U_B(T)\big] = \{1\}$} to[bend left]  (1.7,0);
  	\end{scope}
  	
  	\end{tikzpicture}
  \caption{(Left) $\IE$ compares the density in the upper right quadrant between times $t=0$ and $t=T$. Only the electrons localized at the edge contribute since $\UB(T) = \idB$.  (Right) Bulk and edge spectra. The latter might be gapless but only with extra states confined at the boundary.\label{fig:IE}}
\end{figure}
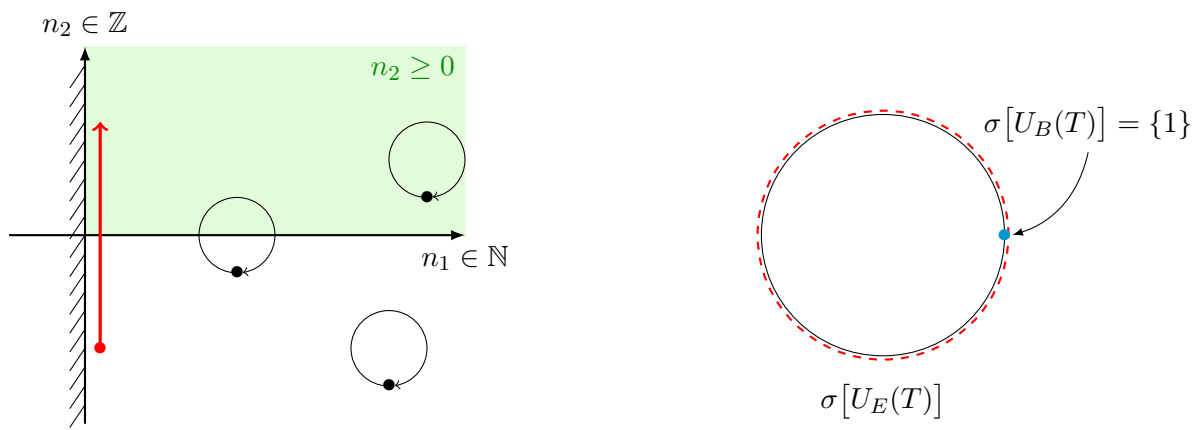

\begin{prop}[Bulk index definition] \label{prop:defI_per}
	 Let $\HB$ be a bulk Hamiltonian and $\UB$ the corresponding propagator such that $\UB(T) = \idB$. The bulk index
	\begin{equation}\label{defIB_per}
	\IB \equiv  \dfrac{1}{2} \int_0^T \dd t\, \Tr_{\HHB} \Big( \UB^*\partial_t \UB \Big[\UB^*[\Lambda_1,\UB],\,\UB^*[\Lambda_2,\UB] \Big]\Big)
	\end{equation}
	is well-defined and independent of the choice of the switch function $\Lambda_i$ in direction $i=1,2$. It is moreover an integer depending continuously on $\HB$ (in local norm) as long as $\UB(T) = \idB$.
\end{prop}

The bulk-edge correspondence then provides a physical interpretation of the bulk index\footnote{It was recently identified with a magnetization density in a particular system \cite{Nathan16}.} through the edge picture. The main result, proved in Sect.\,\ref{sec:proof_BE_per}, is indeed

\begin{thm}[Bulk-edge correspondence]\label{thm:BE_per}
	 Let $\HB$ be a bulk Hamiltonian so that $\UB(T) = \idB$. Then 
	\begin{equation}
	\IB = \IE.
	\end{equation}
\end{thm}

Finally note that this approach generalizes the one from \cite{RudnerPRX13} where translation invariance in space is assumed, namely
\begin{equation}
\HB(t)_{\vm,\vn} = \HB(t)_{0,\vn-\vm} \qquad \forall \vm,\vn \in \mathbb Z^2,
\end{equation}
which is then also true for $\UB(t)$, $\HE(t)$ and $\UE(t)$ (only in direction 2 for the edge operators). We denote by $\widehat \UB(t,k_1,k_2): \mathbb T^3 \mapsto \mathcal U(\mathbb C^N)$ and $\widehat \UE(t,k_2) : \mathbb T^2 \mapsto \mathcal U(\ell^2(\mathbb N) \otimes \mathbb C^N)$ their corresponding Fourier transform.

\begin{prop}[Translation-invariant case]\label{prop:space_periodic_case}
	 Let $\HB$ be a bulk Hamiltonian that is translation invariant, then $\IE$ is the winding number of $\UE(T)$ along $k_2$, i.e.,
	\begin{equation}\label{IE_winding}
	\IE  =\dfrac{\ii}{2\pi}  \int_0^{2\pi} \dd k_2 \, \Tr_{\ell^2(\mathbb N) \otimes \mathbb C^N} \big(\widehat \UE^*(T,k_2) \partial_{k_2}\widehat \UE(T,k_2) \big) 
	\end{equation}
	and $\IB$ is the 3d-winding number of $\widehat \UB$, namely
	\begin{equation} \label{IB_degree}
	\IB = -\dfrac{1}{8\pi^2} \int_{\mathbb T^3} \dd t \dd k_1 \dd k_2 \, \tr_{\mathbb C^N} \big(\widehat \UB^* \partial_t \widehat\UB \big[ \widehat \UB^* \partial_{k_1} \widehat \UB,\, \widehat \UB^* \partial_{k_2} \widehat \UB\big]\big).
	\end{equation}
\end{prop}
Note that a more geometric way to write \eqref{IB_degree} is to use the language of differential forms, namely
\begin{equation}
\IB = -\dfrac{1}{24\pi^2} \int_{\mathbb T^3} \tr_{\mathbb C^N} \big( \big(\widehat \UB^*\, \dd  \widehat \UB\big)^{\wedge 3} \big),
\end{equation}
which is the  degree or odd Chern number \cite{Lyon15,ProdanSchulz16}. Finally a disordered system has been considered in \cite{TitumPRX16} where a finite sample is threaded by fluxes whose parameter space is a torus. That torus replaces the Brillouin zone of the space-periodic case. Thus the expression of the bulk invariant there is analogue to \eqref{IB_degree} by averaging over those fluxes, even though it is evident, at least heuristically, that there is no dependence on them in the thermodynamic limit.

\subsection{The general case \label{sec:generalcase}}

In the general case, $\UB(T) \neq \idB$, we shall define a bulk and edge index for each spectral gap of $\UB(T)$ by deforming the latter to $\idB$ and therefore come back to the previous case. Before doing that we establish the bulk-edge correspondence in a more general context. Consider two bulk Hamiltonians $\HBo$ and $\HBt$ together with their respective propagators $\UBo$ and $\UBt$ which are assumed to satisfy
\begin{equation}\label{UB1=UB2}
\UBo(T) = \UBt(T).
\end{equation}

We join the two Hamiltonians to a single one by placing their times intervals back to back, so to speak with opposite arrow of time. Explicitly, we define the relative Hamiltonian as 
\begin{equation}\label{def_HBrel}
H_\mathrm{B,rel}(t) = \left\lbrace \begin{array}{ll}
2\HBo(2t), & (0 < t < T/2) \\
-2\HBt(2(T-t)), &  (T/2 < t <  T)
\end{array}\right. 
\end{equation}
where the rescaling allows to keep the period $T$. By periodicity the second entry can be written more symmetrically to the first one as $-2\HBt(-t)$ for $-T/2 < t < 0$. The Hamiltonians complies with Def.\,\ref{def:bulk_Hamiltonian} despite jump discontinuities at $t=T/2$ and $T$. The corresponding evolution is 
\begin{equation}\label{def_UBrel}
U_\mathrm{B,rel}(t) = \left\lbrace \begin{array}{ll}
\UBo(2t), & (0 \leq t \leq T/2) \\
\UBt(2(T-t)), &  (T/2 \leq t \leq  T)
\end{array}\right.
\end{equation}
with continuity at $t=T/2$ by \eqref{UB1=UB2}. It satisfies $U_\mathrm{B,rel}(T)=1$ as intended. Indeed, the construction from the previous section applies.
\begin{cor}[Relative bulk-edge correspondence]\label{prop:rel_BE}
	 Let $\HBo$ and $\HBt$ be two bulk Hamiltonians such  that $\UBo(T)=\UBt(T)$. Consider the relative Hamiltonian $H_\mathrm{B,rel}$, \textit{cf.}  \eqref{def_HBrel}, and the associated  propagator $U_\mathrm{B,rel}$, as well as $H_\mathrm{E,rel} = \iota^* H_\mathrm{B,rel} \iota$ and $U_\mathrm{E,rel}$. The relative bulk and edge indices, defined by
	\begin{equation}
	\IB^\mathrm{rel} = \IB[U_\mathrm{B,rel}] \qquad \IE^\mathrm{rel} = \IE[U_\mathrm{E,rel}(T)],
	\end{equation}
	satisfy all the properties of Prop.\,\ref{def:IEperiodic} and \ref{defIB_per}, and moreover Thm.\,\ref{thm:BE_per} applies, namely
	\begin{equation}
	\IB^\mathrm{rel} = \IE^\mathrm{rel}.
	\end{equation}
\end{cor}

Given a single bulk Hamiltonian $\HB$, it is still possible to define bulk and edge indices through this relative construction. The required second Hamiltonian $H_0$ will be chosen as time-independent and in such a way that $\UB(T) = \ee^{- \ii T H_0}$, \textit{i.e.} as a logarithm of $\UB(T)$.

\begin{defn}[Effective Hamiltonian]\label{def:Heff}
	 Let $\HB$ be a bulk Hamiltonian and pick $\varepsilon$ so that $\ee^{-\ii T \varepsilon}$ belongs to a gap of $\UB(T)$. The effective Hamiltonian is defined on $\HHB$ by
	\begin{equation}
	\HB^\varepsilon = \dfrac{\ii}{T} \log_{-T\varepsilon} \UB(T)
	\end{equation}
	through spectral decomposition of $\UB(T)$, where $-T\epsilon$ is the branch cut of the logarithm, defined by $
\log_{\alpha}(\ee^{\ii \phi}) = \ii \phi$ for $\alpha-2\pi < \phi < \alpha$.
\end{defn}

It will be shown in Prop.\,\ref{prop:Heff_local} that $\HB^\varepsilon$ is local. It thus conforms with Def.\,\ref{def:bulk_Hamiltonian}, since its other conditions hold true obviously. The pair $\HB,\,\HB^\varepsilon$ satisfy \eqref{UB1=UB2}, so that we have the general result:
\begin{thm}\label{thm:BE}
	(Bulk-edge correspondence) Let $\HB$ be a bulk Hamiltonian and $\varepsilon$ so that $\ee^{-\ii T \varepsilon}$ belongs to a gap of $\UB(T)$. Consider the relative Hamiltonian $H_\mathrm{B,rel}^\varepsilon$, defined by \eqref{def_HBrel} with $\HBo = \HB$ and $\HBt=\HB^{\varepsilon}$ from Def.\,\ref{def:Heff}, and the associated relative operators $U_\mathrm{B,rel}^{\varepsilon}$, $H_\mathrm{E,rel}^{\varepsilon}$ and $U_\mathrm{E,rel}^{\varepsilon}$. The bulk and edge indices
	\begin{equation}\label{IBIErel_def}
	\IB(\varepsilon) = \IB[U_\mathrm{B,rel}^{\varepsilon}], \qquad \IE(\varepsilon) = \IE[U_\mathrm{E,rel}^{\varepsilon}(T)]
	\end{equation}
	satisfy all the properties of Prop.\,\ref{def:IEperiodic} and \ref{defIB_per}, and moreover Thm.\,\ref{thm:BE_per} applies, namely
	\begin{equation}
	\IB(\varepsilon) = \IE(\varepsilon).
	\end{equation}
\end{thm}

This is nothing but a specific case of Cor.\,\ref{prop:rel_BE}: We constructed a relative evolution that fulfills the assumption of the previous section, namely $H_\mathrm{B,rel}^\varepsilon$ is a bulk Hamiltonian and $U_\mathrm{B,rel}^{\varepsilon}(T) = \idB$. The influence of the choice of $\varepsilon$ is summarized by the next two statements:

\begin{lem}\label{lem:Heff_infl_epsilon}
	Let $\HB$ be a bulk Hamiltonian and $\varepsilon,\,\varepsilon' $ so that $\ee^{-\ii T \varepsilon}$ and $\ee^{-\ii T \varepsilon'}$ belong to a gap of $\UB(T)$. Then
	\begin{equation}\label{Heff2pi}
	\HB^{\varepsilon + 2\pi/T} - H_B^{\varepsilon} = \dfrac{2\pi}{T} \idB
	\end{equation} 
	and for $0 \leq \varepsilon'- \varepsilon < 2\pi/T$
	\begin{equation}\label{comp_Heff}
	\HB^{\varepsilon'} - \HB^\varepsilon = \dfrac{2\pi}{T} P_{\varepsilon,\varepsilon'},
	\end{equation}
	where $P_{\varepsilon,\varepsilon'}$ is the spectral projection of $U_B(T)$ associated to the spectrum between $\ee^{-\ii T \varepsilon}$ and  $\ee^{-\ii T \varepsilon'}$ clockwise.
\end{lem}

\begin{prop}
	[Influence of $\varepsilon$]\label{prop:influence_varep} Let $\HB$ be a bulk Hamiltonian and $\varepsilon,\,\varepsilon' $ so that $\ee^{-\ii T \varepsilon}$ and $\ee^{-\ii T \varepsilon'}$ belong to a gap of $\UB(T)$. Then 
	\begin{equation}\label{Ibid}
	\IB(\varepsilon + 2\pi/T) = \IB(\varepsilon)
	\end{equation}
	 and for $0 \leq \varepsilon'-\varepsilon  < 2 \pi/T$
	\begin{equation}\label{Ibchern}
	\IB(\varepsilon') - \IB(\varepsilon) = c(P_{\varepsilon,\varepsilon'}),
	\end{equation}
	where
	\begin{equation}
	c(P) = - 2 \pi \ii\, \Tr \Big( P \big[[\Lambda_1, P],\, [\Lambda_2, P]\big] \Big) \in \mathbb Z
		\end{equation}
	 is the non-commutative Chern number (or Kubo-St\v{r}eda formula \cite{AvronSeilerSimon94}) of $P$.
\end{prop}

We have $\IB(\varepsilon) = \IB(\varepsilon')$ when $\varepsilon$ and $\varepsilon'$ belong to the same gap, by $P_{\varepsilon,\varepsilon'}=0$. Similarly \eqref{Ibchern} implies \eqref{Ibid} by $P_{\varepsilon,\varepsilon'}=I$ if $\varepsilon'\nearrow\varepsilon+2\pi$. Note that \eqref{Heff2pi} also implies \eqref{Ibid} through $\UB^{\varepsilon+2\pi/T}(t)=\UB^{\varepsilon}(t) \ee^{-2\pi \ii t/T}$. In regards to the operator seen in \eqref{defIB_per} that change contributes a commutator, which leaves the trace and hence the index unaffected.

A typical situation is illustrated in Fig.\,\ref{fig:phase_spectrum}: To each gap of $\UB(T)$ one associates a single index $\IB$, and indices between two distinct gaps are related through the Chern number of the band in between, so that the set of Chern numbers only gives the relative value of the gap indices. Finally note that Thm.\,\ref{thm:BE} generalizes \ref{thm:BE_per} since when $\UB(T) = \idB$ then $\HB^\varepsilon = 0$ for every $0<\varepsilon<2\pi$ so that $\IB(\varepsilon)$ coincides with $\IB$ from the previous section.

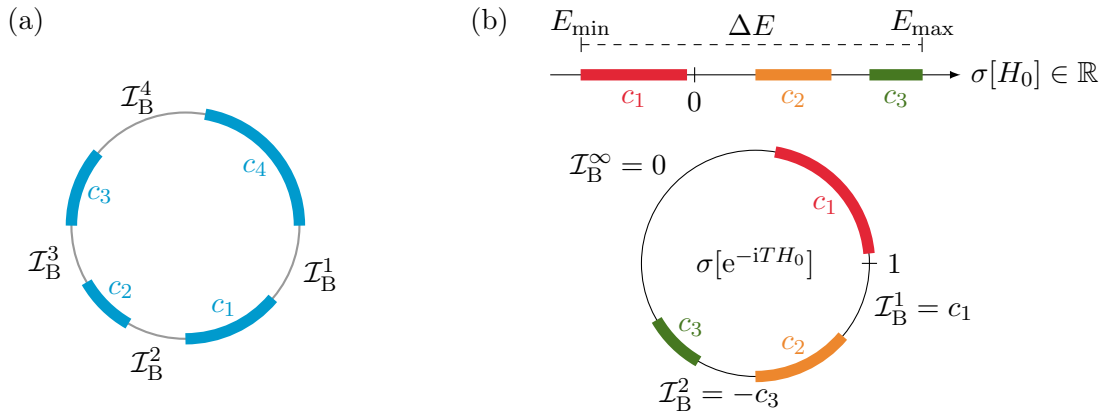
\begin{figure}[htb]
	\centering
	\begin{tikzpicture}[>=latex]
	\newdimen\radius
	\pgfmathsetlength\radius{1.5cm}
	\draw[black!40!white,thick] (0,0) circle (\radius);
	\draw[line width=0.15cm,DeepSkyBlue3] (00:\radius) arc (00:80:\radius);
	\draw[line width=0.15cm,DeepSkyBlue3] ({140}:\radius) arc ({140}:{180}:\radius);
	\draw[line width=0.15cm,DeepSkyBlue3] ({210}:\radius) arc ({210}:{240}:\radius);
	\draw[line width=0.15cm,DeepSkyBlue3] ({270}:\radius) arc ({270}:{320}:\radius);
	\node[DeepSkyBlue3] at (40:1.2) {$c_4$};
	\node[DeepSkyBlue3] at (160:1.2) {$c_3$};
	\node[DeepSkyBlue3] at (225:1.2) {$c_2$};
	\node[DeepSkyBlue3] at (295:1.2) {$c_1$};
	\node at (-20:1.9){$\IB^1$};
	\node at (110:1.8){$\IB^4$};
	\node at (195:1.9){$\IB^3$};
	\node at (255:1.9){$\IB^2$};
	\node at (-1.4\radius,1.8*\radius) {(a)};
	\node at (4,1.8*\radius) {(b)};
	
	\begin{scope}[xshift=7.5cm,yshift=2cm]
	
	\draw[->] (-2.7,0) -- (2.7,0) node[right]{$\sigma[H_0] \in \mathbb{R}$};
	\draw[line width=0.15cm,ct_red] (-0.9,0)-- node[below]{$c_1$} (-2.3,0); 
	\draw[line width=0.15cm,ct_orange] (0,0) --node[below]{$c_2$} (1,0);
	\draw[line width=0.15cm,ct_green]  (1.5,0) -- node[below]{$c_3$} (2.2,0);
	\draw (-0.8,-0.1)node[below]{$0$} -- (-0.8,0.1) ;
	\draw[dashed,|-|] (-2.3,0.4)node[above]{$E_\mathrm{min}$} --node[above]{$\Delta E$} (2.2,0.4)node[above]{$E_\mathrm{max}$};
	\end{scope}		
	\begin{scope}[xshift=7.5cm,yshift=-0.5cm]
	\newdimen\radius
	\pgfmathsetlength\radius{1.5cm}
	\draw (0,0) circle (\radius);
	\draw[line width=0.15cm,ct_red] (05:\radius) arc (05:80:\radius);
	\draw[line width=0.15cm,ct_green] ({210}:\radius) arc ({210}:{240}:\radius);
	\draw[line width=0.15cm,ct_orange] ({270}:\radius) arc ({270}:{320}:\radius);
	\draw (1.4,0) -- (1.6,0) node[right]{$1$};
	\node at (0,0) {$\sigma[\ee^{-\ii T H_0}] $};
	\node at (145:2.2) {$\IB^\infty = 0 $};
	\node at (-15:2.3) {$\IB^1 = c_1 $};
	\node at (-105:1.8) {$\IB^2 = -c_3 $};
	\node[ct_red] at (40:1.2) {$c_1$};
	\node[ct_green] at (225:1.2) {$c_3$};
	\node[ct_orange] at (295:1.2) {$c_2$};
	\end{scope}
	\end{tikzpicture}
	\caption{\label{fig:phase_spectrum}(a) Example of spectrum of $\UB(T)$ with gap indices $\IB^i$ related by Chern numbers $c_i$ of the bands through $\IB^{i+1} -\IB^i = c_i$. (b) For a time-independent Hamiltonian, the index of the \enquote{gap at infinity} always vanishes so that gap indices and Chern numbers are equivalent here. }
\end{figure}

\begin{rem}
	For a time-independent bulk Hamiltonian $H_0$, the set of gap invariants is equivalent to the set of Chern numbers, as illustrated in Fig.\,\ref{fig:phase_spectrum}(b). Indeed the spectrum of $\ee^{-\ii T H_0}$ consists in winding the spectrum of $H_0$ around the unit circle. Here $T$ is arbitrary but as long as $T <(2\pi) / \Delta E$
	where $\Delta E$ is the bandwidth of $H_0$, the propagator $\ee^{-\ii T H_0}$ possesses a \enquote{gap at infinity} coming from the gluing of the trivial gaps of $H_0$ at $\pm \infty$. When taking the branch cut in this gap (e.g. at $\varepsilon =  E_\mathrm{min} - \eta$ for $\eta$ small enough), one has $\HB^\varepsilon = H_0$, so that $\UBo = \UBt$ in the relative evolution. A direct computations shows that $\IB^\infty =  0$
	which sets a reference value for $\IB$. Thus the other gap indices are in one-to-one correspondence with the set of Chern numbers by \eqref{Ibchern}.
\end{rem}

Finally this construction is stable under continuous deformations.

\begin{cor}[Homotopy invariance]\label{cor:IB_contiuity}
	Let $\HBz$ and $\HBo$ be two bulk Hamiltonians related by a homotopy $\HBs$ of bulk Hamiltonians for $s \in [0,1]$. Assume the existence of $\varepsilon $ so that for every $s$, $\ee^{-\ii T \varepsilon}$ belongs to a gap of $U_{\mathrm{B},s}(T)$. Then 
	\begin{equation}
	\mathcal I_{\mathrm B,0}(\varepsilon) = \mathcal I_{\mathrm B,1}(\varepsilon) \tand \mathcal I_{\mathrm E,0}(\varepsilon) = \mathcal I_{\mathrm E,1}(\varepsilon)
	\end{equation}
\end{cor}

The proof of it follows from the continuity of the indices, and the fact that $\UB \mapsto \HB^\varepsilon$ is continuous, see Prop.\,\ref{HeffcontinousinUB}.

\begin{rem}[Weak disorder] \label{rem:disorder}
	 Any disordered configuration can be implemented through $\HB$ and the  construction works as long as a spectral gap is open. Moreover Thm.\,\ref{thm:BE} is deterministic in the sense that the definition of $\IB$ and $\IE$ and the bulk-edge correspondence are valid for any configuration and do not rely on ergodicity or average computation. Finally the indices are continuous in $\HB$ in the sense of Rem.\,\ref{rem:topology}, so that they coincide for two close configurations.
	 
	For example this covers the model developed in \cite{TitumPRX16}, but more generally 
	take $\HB^\omega= H_0 + \lambda V^\omega$ with $H_0$ translation invariant, $\{V^\omega\}_{\omega \in \Omega}$ a random potential and small $\lambda$ so that, from $0$ to  $\lambda$, $\ee^{-\ii T \varepsilon}$ is in a spectral gap of $\UB^{\omega}(T)$. Then $\IB^\omega(\varepsilon) = \IE^\omega(\varepsilon)$ for any $\omega \in \Omega$ and $\IB^\omega(\varepsilon) = \IB(\varepsilon)$, the latter corresponding to $\lambda =0$.
\end{rem}

\subsection{Index of an interface \label{subsec:interface} and space-time duality}

Though the invariants $\IB(\varepsilon)$ and $\IE(\varepsilon)$ of Thm.\,\ref{thm:BE} are mathematically well-defined and coincide, the physical interpretation of the relative evolution and effective Hamiltonian are not obvious. Here we propose a more intuitive reformulation by replacing the edge with an interface.

Consider again the general relative evolution of Cor.\,\ref{prop:rel_BE}. For two bulk Hamiltonians $\HBo$ and $\HBt$ such that $\UBo(T) = \UBt(T)$ we define the relative Hamiltonian $H_\mathrm{B,rel}$ and deduce $U_\mathrm{B,rel}$, $H_\mathrm{E,rel} = \iota^* H_\mathrm{B,rel} \iota$ and $U_\mathrm{E,rel}$. In particular 
\begin{equation}
U_\mathrm{E,rel}(t)=	\left\lbrace \begin{array}{lll}
U_{\mathrm E,1}(2t), & (0 \leq t \leq T/2) \\
U_{\mathrm E,2}(2(T-t))U_{\mathrm E,2}^*(T)U_{\mathrm E,1}(T), &  (T/2 \leq t \leq  T)
\end{array}\right.
\end{equation}
with $U_\mathrm{E,rel}(T)= U_{\mathrm E,2}^*(T)U_{\mathrm E,1}(T)$, so that the  edge index of Cor.\,\ref{prop:rel_BE} can be reformulated as, \textit{cf.}  \eqref{IE_def},
\begin{equation}\label{IErel}
\IE^\mathrm{rel} = \Tr_{\HHE} \Big( \big[\Lambda_2, U_{\mathrm E,1}(T)\big] U_{\mathrm E,1}^*(T) -\big[ \Lambda_2 ,\,U_{\mathrm E,2}(T) \big] U_{\mathrm E,2}^*(T)   \Big).
\end{equation}
The expression looks like the difference of two edge indices from Prop.\,\ref{def:IEperiodic} except that the trace cannot be split since $\UBo(T) = \UBt(T)$ differ from $\idB$. However that suggests:
\begin{defn}[Interface Hamiltonian]\label{def:interfaceH}
	 Let $\HBo$ and $\HBt$ be two bulk Hamiltonians and $H_\mathrm{int}$ be a bulk Hamiltonian that also satisfies $|H_\mathrm{int}(t)_{\vm,\vn}| \leq C \ee^{-\mu|n_1|}\ee^{-\mu|m_2-n_2|}$. Then define
	\begin{equation}
	\HI(t) = P_1 \HBo(t) P_1 + (1-P_1) \HBt(t) (1-P_1) + H_\mathrm{int}(t).
	\end{equation}
\end{defn}

This interface Hamiltonian is a bulk Hamiltonian acting on $\HHB$ and gluing $\HBo$ and $\HBt$ on each half of the sample through a perturbation $H_\mathrm{int}$ confined to the interface, as illustrated in Fig.\,\ref{fig:interface}:

\begin{figure}[htb]
	\centering
	\begin{tikzpicture}
	\draw[-latex] (-5,0) -- (5,0) node[right]{$n_1 \in \mathbb Z$};
	\draw (0,0.1) -- (0,-0.1) node[below]{$0$};
	\draw (-3.5,0) node[above]{$H_{B2}$}; \draw (3.5,0) node[above]{$H_{B1}$};
	\draw [dotted,domain=-3:3,scale=0.7,samples=50] plot(\x,{exp(-\x*\x)});
	\draw (0,0) node[above]{$H_\mathrm{int}$};
	\end{tikzpicture}
	\caption{Interface between two samples. \label{fig:interface}}
\end{figure}
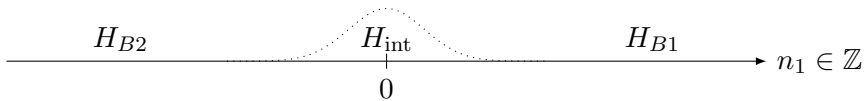

\begin{prop}[Interface index]\label{prop:interfaceindex}
	 Let $\HBo$ and $\HBt$ be two bulk Hamiltonians such that $\UBo(T)=\UBt(T)\equiv\UB(T)$. Consider the interface Hamiltonian $\HI$ from Def.\,\ref{def:interfaceH} and its evolution $\UI$. Then the interface index, defined by
	\begin{equation}\label{defII}
	\II = \Tr_{\HHB} \Big(\UI^*\UB(T)\Big[\Lambda_2, \UB^*\UI(T)\Big]\Big),
	\end{equation}
	is well-defined, integer valued, independent of the choice of $\Lambda_2$ and independent of $H_\mathrm{int}$. Moreover
	\begin{equation}
	\II = \IE^\mathrm{rel},
	\end{equation}
	where $\IE^\mathrm{rel}$ is the relative edge index \eqref{IBIErel_def}  associated to $\HBo$ and $\HBt$.
\end{prop}

Prop.\,\ref{prop:interfaceindex} establishes a duality between space and time. It tells that the relative index $\IE^\mathrm{rel}$ (and consequently $\IB^\mathrm{rel}$ through the bulk-edge correspondence) is nothing but the index $\II$ of a sharp interface between two samples ruled by $H_{\mathrm{E},1}$ and $H_{\mathrm{E},2}$, as illustrated in Fig.\,\ref{fig:duality}. Moreover a smooth gluing through $H_\mathrm{int}$ confined around the interface leads to the same index.

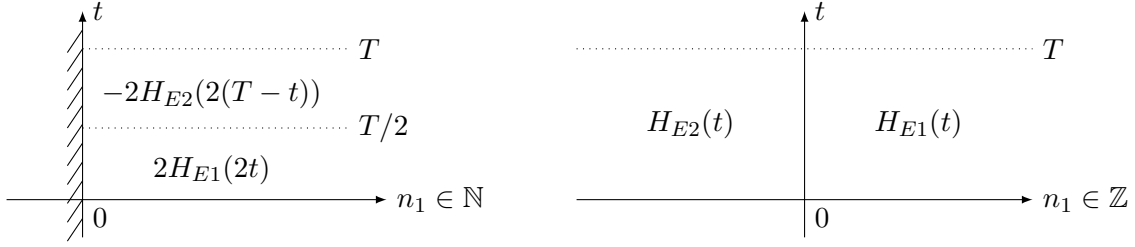
\begin{figure}[htb]
	\centering
	\begin{tikzpicture}
	\draw[-latex] (-1,0) -- (4,0) node[right]{$n_1 \in \mathbb N$};
	\draw[-latex] (0,-0.5)node[above right]{$0$} -- (0,2.5) node[right]{$t$};
	\foreach \k in {-1,...,9}
	{\draw (0,2.5/10*\k) -- (-0.2,2.5/10*\k-0.3);}
	\draw (1.7,1.4) node{$-2 H_{E2}(2(T-t))$}; \draw (1.7,0.4) node{$2 H_{E1}(2t)$};
	\draw[dotted] (0,0.95) -- (3.5,0.95) node[right]{$T/2$};
	\draw[dotted] (0,2) -- (3.5,2) node[right]{$T$};
%	\draw (-1.2,1) node{$H=0$};
	
	\begin{scope}[xshift=9.5cm]
	\draw[-latex] (-3,0) -- (3,0) node[right]{$n_1 \in \mathbb Z$};
	\draw[-latex] (0,-0.5)node[above right]{$0$} -- (0,2.5) node[right]{$t$};
	\draw (1.5,1) node{$ H_{E1}(t)$};
	\draw[dotted] (-3,2) -- (3,2) node[right]{$T$};
	\draw (-1.5,1) node{$ H_{E2}(t)$};
	\end{scope}
	\end{tikzpicture}

	\caption{Duality of space and time. Left: Relative evolution on a sample with an edge. Right: Evolution on the interface between two samples. \label{fig:duality}}
\end{figure}

\begin{rem}
	[Effective vacua] In the context of Thm.~\ref{thm:BE}, where $\HBo=\HB$ and $\HBt=\HB^\varepsilon$, we deduce $\IB(\varepsilon) = \IE(\varepsilon) = \II(\varepsilon)$. So the bulk index counts the number of topological edge modes appearing at the interface between the original and an effective sample ruled by $\HB^\varepsilon$. Hence the latter plays the role of a vacuum that selects the gap of $\UB(T)$ around $\ee^{- \ii T \varepsilon}$, in analogy with the choice of Fermi energy. 
	This vacuum depends on the system but is described by a time-independent and local dynamics, and there are as many distinct vacua as gaps in $\UB(T)$. By expanding the commutators in \eqref{IErel} we have
	\begin{equation}
	\IE(\varepsilon) = \Tr_{\HHE} \Big(\UE(T) \Lambda_2 \UE^*(T) - \ee^{-\ii T \HE^\varepsilon} \Lambda_2 \ee^{\ii T \HE^\varepsilon} \Big),
	\end{equation}
	so that the interpretation of Fig.\,\ref{fig:IE} still holds: $\IE(\varepsilon)$ measures the charge pumped in the upper quadrant, but relatively to the dynamics of the effective vacuum. This ensures that the index is well-defined and the pumping remains quantized. Finally if $\UB(T) = \idB$ then $\HB^\varepsilon = 0$ for every $0 < \varepsilon <2 \pi$ so the only effective vacuum is the usual one.
\end{rem}

\section{Properties of bulk and edge propagator \label{sec:timeevol}}

The indices are defined through trace expressions involving $\UB$ or $\UE$. In this section we study their properties and compare them. Before we recall a series of lemmas relating local and trace class operators through the notion of confinement \cite{ElgartGrafSchenker05}.

\subsection{Locality, confinement and switch functions \label{sec:local_conf_trace}}

In the following we say that $f : \mathbb Z^2$ or $\mathbb N \times \mathbb Z\rightarrow \mathbb R$ is a Lipschitz function (of constant 1) if it satisfies 
\begin{equation}\label{def_1Lip}
 \qquad |f(\vm)-f(\vn)| \leq |\vm - \vn|,\qquad \forall\, \vm,\vn \in \mathbb Z^2
\end{equation}
For $\lambda >0$ we denote by $\ee^{\lambda f}$ and $\ee^{-\lambda f}$ the multiplicative operators on $\HHB$ or $\HHE$.  We shall first rephrase the notion of locality appearing in Def.\,\ref{def:bulk_Hamiltonian}.

\begin{lem}\label{lem:H_expdecay}
	Let $A$ be a local operator on $\HHB$ or $\HHE$ with locality exponent $\mu>0$, then
	\begin{equation}
	\norm{A} \leq \norm{A}_\mu c(\mu)
	\end{equation}
	with $c(\mu) < \infty$ and 
	\begin{equation}
	 \forall\, 0 \leq \lambda < \mu \qquad  \Vert\ee^{- \lambda f} A \ee^{\lambda f}-  A\Vert \leq \norm{A}_\mu b(\lambda)  < \infty 
	\end{equation}
	for any Lipschitz function $f$ where $b(\lambda) \rightarrow 0$, $(\lambda \rightarrow 0)$.
\end{lem}
\begin{proof}
	We apply the Holmgren-Schur estimate
	\begin{equation}
	\lVert A \rVert \leq |||A||| \equiv \max \Big( \sup_{\vm \in \mathbb Z^2} \sum_{\vn\in \mathbb Z^2} |A_{\vm,\vn}| \, , \, \vm \leftrightarrow \vn \Big)
	\end{equation}
	and estimate $|||A||| \leq C c(\mu)$ with $c(\mu) = \sum_{\vn \in \mathbb Z^2} \ee^{-\mu |\vn|}$ for any $C$ as in the definition \eqref{local_norm} of the local norm. We then pass to the infimum over $C$. As for the second inequality we estimate
	\begin{align}
	 |(\ee^{- \lambda f} A \ee^{\lambda f}- A)_{\vm,\vn}| &  =  \big|\ee^{-\lambda f(\vm)}A_{\vm,\vn}\ee^{\lambda f(\vn)}- A_{\vm,\vn}\big| \cr
	& \leq  C \ee^{- \mu |\vm-\vn |}\big( \ee^{\lambda |\vm-\vn |}-1\big) 
	%& \leq \norm{O}_\mu \sum_{\vn\in \mathbb Z^2}  \ee^{- \mu |\vm-\vn |} \big( \ee^{\lambda |\vm-\vn|}-1\big)  \equiv \norm{O}_\mu b(\lambda)
	\end{align}
	where we have used  $|\ee^{a} - 1 | \leq \ee^{|a|}-1$ and \eqref{def_1Lip}, so as to obtain 
	\begin{equation}
	||| \ee^{- \lambda f} A \ee^{\lambda f}- A ||| \leq \norm{A}_\mu b(\lambda)
	\end{equation}
	with $b(\lambda) = \sum_{\vn\in \mathbb Z^2}  \ee^{- \mu |\vn |} \big( \ee^{\lambda |\vn|}-1\big) < \infty$. Finally $b(\lambda) \rightarrow 0$, ($\lambda \rightarrow 0$) by dominated convergence.
\end{proof}

\begin{cor}\label{cor:localbis}
	Let $A$ be a local operator on $\HHB$ or $\HHE$ with locality exponent $\mu$, then
	\begin{equation}\label{local2}
	\forall\, 0 \leq \lambda < \mu \qquad  \norm{\ee^{- \lambda f} A \ee^{\lambda f}} \leq B_\lambda \norm{A}_\mu 
	\end{equation}
	with $B_\lambda>0$ and for any Lipschitz function $f$. Conversely if $\ee^{- \lambda f} A \ee^{\lambda f}$ is bounded for some $\lambda>0$ and any $f$ Lipschitz then it is local, namely $\norm{A}_\lambda \leq \sup_f\norm{\ee^{-\lambda f}A\ee^{\lambda f}}$.
\end{cor}
\begin{proof}
	The first statement is an immediate consequence of the previous lemma and the triangle inequality. The second one is proved by taking for any $\vm,\vn \in \mathbb Z^2$, the Lipschitz function $f(\vp) = |\vp - \vm|$ leading to
	\begin{align}
	|A_{\vm,\vn}| &= | \langle \vm  |\ee^{\lambda f} \ee^{-\lambda f}   A\ee^{\lambda f} \ee^{-\lambda f} | \vn \rangle| 
	\leq \norm{\ee^{- \lambda f} A \ee^{\lambda f}} \ee^{-\lambda |\vm-\vn|},
	\end{align} 
	where we have used Cauchy-Schwarz inequality and the fact that $|m\rangle$ and $|n\rangle$ are normalized. 
\end{proof}

Thus locality of $A$ as by Rem.\,\ref{rem:topology} is equivalent to the boundedness of $\ee^{- \lambda f} A \ee^{\lambda f}$, up to a change from $\mu$ to $\lambda < \mu$. We then refine this notion by considering Lipschitz functions in direction $i$ as Lipschitz functions $f_i(\vn) =f(n_i)$. We observe that $f_1 + f_2$ is again Lipschitz, \textit{i.e.} with constant 1.

\begin{defn}
	A bounded operator $A$ on $\HHB$ or $\HHE$  is called, for $i,j = 1,2$, $i\neq j$:
	\begin{itemize}
		\item $i$-\textit{local} (or \textit{local in direction $i$}) if it exists $\lambda >0$ so that
		$\ee^{-\lambda f_i} A \ee^{\lambda f_i}$
		is bounded for any Lipschitz function $f_i$ in direction $i$.
		\item $i$-\textit{confined} (or \textit{confined in direction $i$}) if it exists $\lambda >0$ so that
		$A \ee^{\lambda |n_i|}$
		is bounded.
		\item simultaneously $i$-confined and $j$-local if it exists $\lambda >0$ so that
		$\ee^{-\lambda f_j} A\ee^{\lambda f_j} \ee^{\lambda |n_i|}$ is bounded.
	\end{itemize}
	The bounds are meant uniformly in $f_i$. The suprema provide norms associated with each property.
\end{defn}

For example, in analogy with Lem.\,\ref{lem:H_expdecay} and Cor.\,\ref{cor:localbis}, equation \eqref{pres_D} for $D$ and Def.\,\ref{def:interfaceH} of $H_\mathrm{int}$ means that these two operators are simultaneously 1-confined and 2-local. Another way to produce $i$-confined operators is to use switch functions from Def.\,\ref{def:switch}.

\begin{lem}\label{lem:Lambdaconfinement}
	Let $A$ be a local operator and $\Lambda_i$ a switch function in direction $i$. Then $[\Lambda_i, A]$ is simultaneously $i$-confined and $j$-local with corresponding norm bounded in terms of $\norm{A}_\mu$.
\end{lem}

\begin{proof}
	We rewrite 
	\begin{align}\label{compute_lem_icjl}
	\ee^{-\lambda f_j}[\Lambda_i,A]\ee^{\lambda f_j}\ee^{\lambda |n_i|} 
	= \ee^{-\lambda f_j}\big(\Lambda_i A (1-\Lambda_i) - (1-\Lambda_i)A \Lambda_i\big) \ee^{\lambda f_j}\ee^{\lambda |n_i|}
	\end{align}
	 The first term reads
	\begin{equation}
	\ee^{-\lambda f_j} \Lambda_i A (1-\Lambda_i) \ee^{\lambda f_j} \ee^{\lambda |n_i|}   = \Lambda_i \ee^{-\lambda n_i} \cdot \ee^{-\lambda(f_j-n_i)} A \ee^{\lambda(f_j-n_i)} \cdot (1-\Lambda_i) \ee^{\lambda(n_i+|n_i|)}
	\end{equation}
	
	The middle factor is bounded by \eqref{local2} and the other two are by the support property of the switch function. The second term in \eqref{compute_lem_icjl} is similarly bounded. 
\end{proof}

\begin{lem}\label{lem_icjl_traceclass}
	For $i=1,2$ and $j\neq i$ let $A_i$ be simultaneously $i$-confined and $j$-local, and $\Lambda_j$ a switch function in direction $j$. Then 
$
	[\Lambda_j, A_i] $ and $  A_i A_j 
$
	are trace class with matching bonds on the norms. 
\end{lem}

\begin{proof}
	Similarly to the previous proof we write
	\begin{equation}
	[\Lambda_j, A_i] \ee^{\lambda|n_i|} \ee^{\lambda |n_j|} = \Lambda_j  A_i  \ee^{\lambda|n_i|} (1-\Lambda_j) \ee^{\lambda |n_j|} - (1-\Lambda_j)  A_i  \ee^{\lambda|n_i|} \Lambda_j \ee^{\lambda |n_j|}
	\end{equation}
	and see that each term is bounded by the same arguments. Then we write
	\begin{equation}
	[\Lambda_j, A_i] = [\Lambda_j, A_i] \ee^{\lambda|n_i|} \ee^{\lambda |n_j|} \cdot \ee^{-\lambda|n_i|} \ee^{-\lambda |n_j|}.
	\end{equation}
	The second factor is obviously trace class, and so is $[\Lambda_j, A_i]$. Similarly
	\begin{equation}
	A_i A_j  \ee^{\lambda|n_i|} \ee^{\lambda |n_j|} = A_i\ee^{\lambda|n_i|} \cdot \ee^{-\lambda|n_i|} A_j  \ee^{\lambda|n_i|} \ee^{\lambda |n_j|}
	\end{equation}
	is bounded, so that $A_i A_j$ is trace class.
\end{proof}
This lemma, combined with the previous one, will be of particular interest when $A_i = [\Lambda_i, A]$ for $A$ local or $A_i = D$ from Prop.\,\ref{prop:defD} below. We finally need:
\begin{lem}
	\label{lem:com_switch_0}
	Let $\Lambda_i$ be a switch function in direction $i=1,2$ and $A$ an operator such that $[\Lambda_i, A]$ is trace class. Then
	\begin{equation}
	\Tr \big([\Lambda_i, A]\big) = 0.
	\end{equation}
\end{lem}
\begin{proof}
	Note that only $[\Lambda_i, A]$ is trace class so that we cannot open the commutator and separate the traces. However this allows to compute the trace through the diagonal kernel (take $\mathcal H = \ell^{2}(\mathbb Z^2)$ for concreteness)
	\begin{equation}
	\Tr \big([\Lambda_i, A]\big) = \sum_{\vn \in \mathbb Z^2} ([\Lambda_i, A])_{\vn,\vn} =\sum_{\vn \in \mathbb Z^2} \Lambda(n_i) A_{\vn,\vn} - A_{\vn,\vn}\Lambda(n_i) = 0.
	\end{equation}
	This result  is the analogue of a vanishing integral of a total derivative.
\end{proof}

\subsection{Locality and continuity of propagators \label{sec:localU}}

In this section we show that the propagator is local when the Hamiltonian is, and that the propagator is continuous in the Hamiltonian with respect to the local norm.

\begin{prop}\label{prop:U_expdecay}
	Let $\HB$ be a bulk Hamiltonian in the sense of Def.\,\ref{def:bulk_Hamiltonian}, $\mu$ its locality exponent, and $\UB$ the corresponding unitary propagator. Then
	\begin{equation}
	\forall \, t \in [0,T] \quad \forall\, 0 \leq \lambda < \mu \qquad \norm{\ee^{- \lambda f} \UB(t) \ee^{\lambda f}- \UB(t)} \leq  \alpha_\lambda
	\end{equation}
	for any Lipschitz function $f$ on $\mathbb Z^2$, and with $\alpha_\lambda$ independent of $t$ and $\alpha_\lambda \rightarrow 0$, ($\lambda \rightarrow 0$).
\end{prop}
\begin{proof}
	In order to work with bounded operators we define for a Lipschitz function $f$ its bounded version $f_n = n f/(n+|f|)$ for $n \in\mathbb N$, which is again Lipschitz and so that $f_n \rightarrow f$ when $n \rightarrow \infty$. Define 
	\begin{equation}
	V_n(t) = \ee^{-\lambda f_n} \UB(t) \ee^{\lambda f_n} - \UB(t),
	\end{equation}
	which is bounded for every $n \in \mathbb N$, with $V_n(0)=0$. Thus
	\begin{equation}
\norm{V_n(t)} \leq \int_0^t \norm{\dfrac{\dd V_n(s)}{\dd s}} \dd s.
 	\end{equation}
	 Using that $\UB(t)\varphi$ satisfies Schr\"odinger equation \eqref{Schrodinger_equation} for every $\varphi \in \HHB$,
	we deduce
	\begin{align}
	\norm{\dfrac{\dd V_n(s)}{\dd s}}
	 & = \left\lVert \ee^{-\lambda f_n} \HB(s)\UB(s) \ee^{\lambda f_n} - \HB(s) \UB(s) \right\rVert \cr
	& = \left\lVert \big( \ee^{-\lambda f_n} \HB(s) \ee^{\lambda f_n} - \HB(s) \big)\ee^{-\lambda f_n}\UB(s)\ee^{\lambda f_n} \right.\cr
	&\quad\left.+ \HB(s) \big(\ee^{-\lambda f_n}\UB(s)\ee^{\lambda f_n} - \UB(s) \big) \right\rVert. %\cr
	\end{align}
	Since $\HB$ is local it is bounded uniformly in time by Lem.\,\ref{lem:H_expdecay} and $\norm{\HB(s)} \leq \norm{\HB}_\mu c(\mu)$, 
	we furthermore deduce
	\begin{align}
	\norm{V_n(t) }
	& \leq \int_0^t \big(\norm{\HB}_\mu (c(\mu)+b(\lambda)) \norm{V_n(s)} + \norm{\HB}_\mu   b(\lambda) \big) \dd s,
	\end{align}
	where $c(\mu)$ and $b(\lambda)$ are independent of $t$ and $n$ and $b(\lambda) \rightarrow 0$ when $\lambda \rightarrow 0$. We then get by application of Gr\"onwall Lemma and the fact that $e^x-1\leq x \ee^x$ for $x>0$
	\begin{align}
	\Vert V_n(t) \Vert &\leq \norm{\HB}_\mu   b(\lambda) t \ee^{\norm{\HB}_\mu (c(\mu) +b(\lambda)) t} \cr
	 &\leq \norm{\HB}_\mu  b(\lambda) T \ee^{\norm{\HB}_\mu (c(\mu) +b(\lambda)) T}  \equiv  \alpha_\lambda
	\end{align}
	uniformly in $t \in [0,T]$ and $n \in \mathbb N$. Thus $V = \ee^{-\lambda f} \UB \ee^{\lambda f} - \UB$ is well-defined and $V_n(t)$ converges strongly to $V(t)$. In particular $\Vert V(t) \Vert \leq \alpha_\lambda$, and from the previous expression we see that $\alpha_\lambda \rightarrow 0$ when $\lambda \rightarrow 0$.
\end{proof}

By Cor.\,\ref{cor:localbis}, this last proposition shows that $\UB$ is local for any exponent $\lambda<\mu$. Furthermore we have:

\begin{prop}\label{cor:U_continuity} Let $\HB$ be a bulk Hamiltonian in the sense of Def.\,\ref{def:bulk_Hamiltonian}, with locality exponent $\mu$. Then the map $\HB \mapsto \UB$ is continuous (uniformly in time) with the respective norms $\norm{\cdot}_\mu$ and $\norm{\cdot}_\lambda$ for $\lambda < \mu$.
\end{prop}
\begin{proof}
	Let $\HBo$ and $\HBt$ be two bulk Hamiltonians and $\UBo$, $\UBt$ their propagators. For $\lambda<\mu$ and $f$ a Lipschitz function we compute
	\begin{align}
	\norm{\ee^{-\lambda f} (\UBo-\UBt)(t) \ee^{\lambda f}} &\leq \int_0^t \dd s \norm{\ee^{-\lambda f} (\HBo\UBo-\HBt\UBt)(s) \ee^{\lambda f}}\cr
	& \leq \int_0^t \dd s \Big( B \norm{\HBo}_\mu \norm{\ee^{-\lambda f} (\UBo-\UBt)(s) \ee^{\lambda f}} \cr
	&\quad+ B\norm{\HBo-\HBt}_\mu \norm{\ee^{-\lambda f} \UBt(s) \ee^{\lambda f}} \Big)
	\end{align}
	where we have used Cor.\,\ref{cor:localbis} for $\HBo$ and $\HBo-\HBt$ with some constant $B >0$. Adding $\UBt-\UBt$ in the last norm, by Prop.\,\ref{prop:U_expdecay} and Gr\"onwall inequality as in the previous proof we deduce
	\begin{equation}
	\norm{\ee^{-\lambda f} (\UBo-\UBt)(t) \ee^{-\lambda f}} \leq B'\norm{\HBo-\HBt}_\mu (1+ \alpha_\lambda) T \ee^{B \norm{\HBo}_\mu T}
	\end{equation}
	uniformly in time so that $\norm{\UBo-\UBt}_\lambda$ is similarly bounded.
\end{proof}

Finally note that all this discussion naturally extends to the edge Hilbert space $\HHE$ since $\HE$ satisfies similar properties of $\HB$ by construction, thus all the proofs remain unchanged.

\begin{cor}\label{cor:UE_expdecay}
	Let $\HB$ be a bulk Hamiltonian in the sense of Def.\,\ref{def:bulk_Hamiltonian}, $\HE$ the corresponding edge Hamiltonian defined in \ref{def:edge_Hamiltonian} and $\UE$ the associated propagator. Then Lem.\,\ref{lem:H_expdecay} holds for $H_E$ and Prop.\,\ref{prop:U_expdecay} and \ref{cor:U_continuity} hold for $U_E$.
\end{cor}

\subsection{Comparing bulk and edge propagator}

As announced in the beginning of Sect.\,\ref{sec:main_results}, we prove identity \eqref{pres_D} that compares the edge propagator and the truncated bulk one, and is crucial for the proof of bulk-edge correspondence.

\begin{prop}\label{prop:defD}
	Let $\HB$ and $\HE$ be a bulk and its corresponding edge Hamiltonian, and $\UB$ and $\UE$ the respective propagators. Define 
	\begin{equation}
	D(t) = U_E(t) - \iota^* U_B(t) \iota.
	\end{equation}
	Then for $0 \leq \lambda < \mu$ where $\mu$ is the locality exponent of $\HB$, $D(t)$ is 1-confined and 2-local on $\HHE$, uniformly in $t \in[0,T]$. Moreover the map $\HB \mapsto D(t)$ is continuous in the relevant norms.
\end{prop}

\begin{proof}
	As in proof of Prop.\,\ref{prop:U_expdecay} we start with bounded version of Lipschitz functions, namely $f_2^p = p f_2/(p+|f_2|)$, and similarly $f_1^p = p n_1 / (p+ n_1)$. Then for $t \in [0,T]$ we compute
	\begin{align}\label{F2F1D_comput}
	\ee^{-\lambda f_2^p} &D(t) \ee^{\lambda f_2^p} \ee^{\lambda f_1^p} \cr 
	& =\ee^{-\lambda f_2^p} \big(\UE(t) - \iota^* \UB(t) \iota\big) \ee^{\lambda f_2^p} \ee^{\lambda f_1^p}  \cr 
	&= -\ee^{-\lambda f_2^p} \UE(t)\int_0^t \dd s \, \partial_s \Big( \UE^*(s) \iota^* \UB(s) \iota \, \ee^{\lambda f_2^p} \ee^{\lambda f_1^p}  \Big) \cr
	& = -\ee^{-\lambda f_2^p} \UE(t)\int_0^t \dd s \,  \UE^*(s) \ii \big(\HE(s) \iota^* - \iota^* \HB(s)\big)  \UB(s)\iota\, \ee^{\lambda f_2^p} \ee^{\lambda f_1^p}  \cr
	& = \ii \int_0^t \dd s \, \ee^{-\lambda f_2^p} \UE(t) \UE^*(s) \iota^* \HB(s) (1-P_{1}) \UB(s) \iota\, \ee^{\lambda f_2^p} \ee^{\lambda f_1^p},
	\end{align}
	where we  used $\partial_s \UB  = -\ii \HB \UB$ and $\partial_s \UE^* = \ii \UE^* \HE$, Def.\,\ref{def:edge_Hamiltonian} of $\HE$ and property  \eqref{J_partialisom}. Then we write
	\begin{equation}
	\iota \ee^{\pm\lambda f_2^p} = \ee^{\pm\lambda f_2^p}  \iota \qquad \ee^{\pm\lambda f_2^p} \iota^* = \iota^* \ee^{\pm\lambda f_2^p},
	\end{equation}
	where on the right hand side of each equation is the natural extension of $\ee^{\pm\lambda f_2^p}$ on $\mathcal \HHB$, that we denote by the same symbol, and similarly for $\ee^{\pm\lambda f_1^p}$ since $f_1^p \rightarrow n_1$ has also a natural extension on $\mathbb Z^2$.
	Thus we can rewrite the integrand as
	\begin{align}\label{Dcont}
	& \ee^{-\lambda f_2^p} \UE(t) \UE^*(s) \iota^* \HB(s) (1-P_{1}) \UB(s) \iota \ee^{\lambda f_2^p} \ee^{\lambda f_1^p} \cr
	& =  \ee^{-\lambda f_2^p} \UE(t)  \ee^{\lambda f_2^p} \, \cdot \, \ee^{-\lambda f_2^p} \UE^*(s)  \ee^{\lambda f_2^p} \, \cdot \, \iota^* \, \ee^{-\lambda f_2^p}  \HB(s) \ee^{\lambda f_2^p}\, \cdot \,  (1-P_{1}) \ee^{\lambda f_1^p} \, \cr
	&\hspace{1cm} \cdot \ee^{-\lambda f_1^p}  \ee^{-\lambda f_2^p} U_B(s) \ee^{-\lambda f_2^p} \ee^{\lambda f_1^p} \, \iota
	\end{align}
	Note that $(1-P_1) \ee^{\lambda f_1^p}$ is bounded uniformly in $p \in \mathbb N$. By Lem.\,\ref{lem:H_expdecay}, Prop.\,\ref{prop:U_expdecay} and Cor.\,\ref{cor:UE_expdecay}, each one of the remaining factor is uniformly bounded in $s \in[0,T]$ and $p \in \mathbb N$.	Then so is the operator appearing on the l.h.s. of \eqref{F2F1D_comput}. Since the estimate is independent of $p$, it strongly converges to $\ee^{-\lambda f_2} D(t) \ee^{\lambda f_2} \ee^{\lambda n_1}$ which is also bounded uniformly in $t \in [0,T]$. More precisely
	\begin{equation}\label{aux3}
	\norm{\ee^{-\lambda f_2} D(t) \ee^{\lambda f_2} \ee^{\lambda n_1}} \leq B T \norm{\UE}_\mu^2 \norm{\HB}_\mu \norm{\UB}_{\mu}
	\end{equation}
	with a constant $B >0$. The continuity follows from \eqref{Dcont} by Cor.\,\ref{cor:UE_expdecay}.
\end{proof}

\section{Proofs \label{sec:proofs}}

We mostly follows the order of statements of Sect.\,\ref{sec:main_results}. First we assume that $\UB(T) = \idB$ and prove the bulk-edge correspondence, then check that it also applies to the relative evolution with effective Hamiltonian in the general case. In between some extra properties are established, such as the invariance under change of boundary condition and additivity property of the bulk index.

\subsection{Periodic unitary propagator \label{sec:proof_BE_per}}

\begin{proof}[\textbf{Proof of Prop.\,\ref{def:IEperiodic}}] 	
 When $\UB(T) =  \idB$, Prop.\,\ref{prop:defD} reduces to $\UE(T)=\idE+D(T)$ with $D(T)$ simultaneously 1-confined and 2-local. In particular
 \begin{equation}\label{UE=Dincomm}
 [\Lambda_2, U_E(T)] = [\Lambda_2, D(T)]
 \end{equation}
 is trace class according to Lem.\,\ref{lem_icjl_traceclass}, so that $\IE$ is well-defined. Then for two switch functions $\Lambda_2$ and $\widetilde \Lambda_2$, the difference of the corresponding indices reads
 \begin{equation}
 \IE - \widetilde \IE = \Tr_{\HHE} \Big( U_E^*(T) \big((\Lambda_2-\widetilde \Lambda_2) D(T)-D(T)(\Lambda_2-\widetilde \Lambda_2 ) \big) \Big).
 \end{equation}
 Since $\Lambda_2 -\widetilde \Lambda_2 $ is compactly supported in direction 2, it is 2-confined, and even obviously simultaneously 1-local and 2-confined. Applying again Lem.\,\ref{lem_icjl_traceclass}, we deduce that
 $(\Lambda_2-\widetilde \Lambda_2) D(T)$ and $D(T)(\Lambda_2-\widetilde \Lambda_2 )$
 are separately trace class, so that we can split the trace into two parts that are actually equal by cyclicity and the fact that $\UE^*(T)$ and $D(T) = \idE - \UE(T)$ commute, so that $\IE = \widetilde \IE$.
 
 In particular we can compute $\IE$ with $\Lambda_2 = P_2$ that is also a projection. Rewriting
 \begin{equation}\label{IEasindex2}
 \IE = \Tr_{\HHE} \big( U_E^*(T) P_2 U_E(T) - P_2 \big) = \mathrm{Ind}(U_E^*(T) P_2 U_E(T), P_2)
 \end{equation}
 we recognize the index of a pair of projections (see \cite{AvronSeilerSimon94}) defined by 
 \begin{equation}
 \mathrm{Ind}(P,Q) = \dim \ker (P - Q - 1) - \dim  \ker (P-Q + 1) \in \mathbb Z.
 \end{equation}
Indeed, when $(P-Q)^{2n+1}$ is trace class for some $n \geq 0$ , then 
\begin{equation}
\forall m \geq n \qquad \Tr((P-Q)^{2m+1} ) = \Tr((P-Q)^{2n+1} ) = \mathrm{Ind}(P,Q).
\end{equation}
In our case, $P = U_E^*(T) P_{2} U_E(T)$, $Q= P_{2}$ and $n=0$ gives \eqref{IEasindex2} which is an integer by the definition above. The continuity of $\HB \mapsto \Tr([P_2,D(T)])$ follows by Prop.\,\ref{prop:defD} and Lem.\,\ref{lem_icjl_traceclass}. It implies that of $\IE$ by \eqref{UE=Dincomm}.
%Finally the continuity comes from the fact that
%\begin{equation}
%\norm{ \big[P_2, D(T)\big] }_1 \leq \norm{\big[P_2, D(T) \ee^{\lambda |n_1| } \big] 
%\ee^{\lambda |n_2|}} \norm{\ee^{-\lambda(|n_1|+|n_2|) } }_1
%\end{equation}
%where $\norm{\cdot}_1$ is the trace norm. Then by inspecting the proof of Lem.\,\ref{lem:Lambdaconfinement} and \ref{lem_icjl_traceclass}, we deduce
%\begin{equation}
%\norm{\big[P_2, D(T)\ee^{\lambda |n_1|}] \ee^{\lambda |n_2|}  } \leq B \norm{\ee^{\lambda n_2} D\ee^{-\lambda n_2} \ee^{\lambda |n_1|} } + B' \norm{\ee^{-\lambda n_2} D\ee^{+\lambda n_2}  \ee^{\lambda |n_1|} }
%\end{equation}
%with some constants $B,\,B'>0$. So that, by \eqref{aux3}
%\begin{equation}\label{aux4}
%\norm{ \big[P_2, \UE(T)\big] }_1 \leq B" \norm{\UE}_\mu^2 \norm{\HB}_\mu \norm{\UB}_{\mu}
%\end{equation}
%which ensures the continuity of $\UE(T) \mapsto \Tr([P_2,\UE(T)])$ in $\norm{\cdot}_\mu$, so that by composition $\IE$ inherit the same property. The continuity in $\HB$ follows from Prop.\,\ref{cor:U_continuity}, Cor.\,\ref{cor:UE_expdecay} and the fact that $\norm{H_E}_\mu \leq \norm{\HB}_\mu$.
\end{proof}

Since $\IE$ is continuous and integer valued, it is constant. As a consequence we have:

\begin{prop}[Influence of the boundary condition]\label{prop:general_BC}
	 Let $\HB$ be a bulk Hamiltonian so that $\UB(T) =  \idB$. Consider the alternative edge Hamiltonian
	\begin{equation}
	\widetilde \HE(t) = \iota^* \HB(t) \iota + \Hbc(t)
	\end{equation}
	with $\Hbc$ a self-adjoint operator on $\HHE$ that is simultaneously 1-confined and 2-local. Let $\widetilde \UE$ and $\widetilde \IE$ the associated propagator and edge index. Then
$
\widetilde \IE = \IE
$
	where $\IE$ corresponds to $\Hbc=0$.
\end{prop}

\begin{proof}First note that $\widetilde \HE$ is still local. Thus adapting the proof of Prop.\,\ref{prop:defD} where we replace $\HE$ by $\HE  + \Hbc$ and using the fact that $\Hbc$ is 1-confined and 2-local, we end up with
\begin{equation}
\widetilde \UE(T) = \idE + \widetilde D(T)
\end{equation}
with $\widetilde D(T)$ 1-confined and 2-local, so that $\widetilde \IE$ is well-defined and shares the properties of Prop.\,\ref{def:IEperiodic}. In particular it is now continuous both in $\HB$ and $\Hbc$ in the relevant norms. Moreover note that $\HE$ and $\widetilde \HE$ are homotopic through
\begin{equation}
H_s(t) \equiv \iota^* \HB(t) \iota + s \Hbc(t)
\end{equation}
for $s \in [0,1]$. By Cor.\,\ref{cor:UE_expdecay} this induces a homotopy $U_\lambda$ from $\UE$ to $\widetilde U_E$ which preserves $\HB$ and $\UB$. Thus $\widetilde \IE = \IE$.
\end{proof}

In the translation invariant case, the bulk Hamiltonian satisfies $(\HB)_{\vm,\vn} (t) = (\HB)_{0,\vn-\vm}(t)$ and induces this property on $\UB$, but also on the edge operators $\HE$ and $\UE$ in direction 2, namely $(\UE)_{\vm,\vn}(t) = (\UE)_{(m_1,0),(n_1,n_2-m_2)}(t)$. In that case we define the Fourier transform of $\UE$ in direction 2 by
\begin{equation}
\widehat \UE(t,k_2) =   \sum_{n_2 \in \mathbb Z} \ee^{-\ii k_2 n_2} (\UE)_{0,n_2}(t).
\end{equation}
Noticing that $\ell^2(\mathbb N \times \mathbb Z) \cong \ell^2(\mathbb N) \otimes \ell^2(\mathbb Z)$, this operator is defined for each pair $(m_1,n_1) \in \mathbb N^2$ so that $\widehat \UE(t,k_2)$ acts on $\mathbb \ell^2(\mathbb N) \otimes \mathbb C^N$. Since $\UE$ is local from Cor.\,\ref{cor:UE_expdecay}, this Fourier transform is well-defined and even smooth in $k_2$.

\begin{lem}\label{lem:periodic}
	Let $A$ and $B$ be bounded and translation invariant operators on $\ell^{2}(\mathbb Z)$ with $\mathcal C^1$-Fourier transform $\widehat A(k)$ and $\widehat B(k)$, and $\Lambda$ a  switch function. If $A [\Lambda, B]$ is trace class 
	then
	\begin{equation}
	\Tr_{\ell^2(\mathbb Z)} \Big( A [\Lambda, B] \Big) = (A X B)_{00} = \ii \int_0^{2\pi} \dfrac{\dd k}{2\pi} \widehat A(k) \partial_k \widehat B(k),
	\end{equation}
	where $X$ is the position operator.
\end{lem}
\begin{proof}
	Since $A [\Lambda, B]$ is trace class, its trace can be computed through its diagonal kernel: 
	\begin{align}
	\Tr_{\ell^2(\mathbb Z)} \Big( A [\Lambda, B] \Big) = \sum_{p,q \in \mathbb Z} A_{p,q} B_{q,p} (\Lambda(q) - \Lambda(p) )
	= \sum_{p' \in \mathbb Z} A_{0,p'} p' B_{p',0}
	\end{align}  
	which gives the first equality. We have used $A_{p,q} = A_{0,q-p}$, similarly for $B$, the change of variables $q \mapsto p' = q-p$ and the fact that 
	\begin{equation}
	\sum_{p \in \mathbb Z} (\Lambda(p' + p) - \Lambda(p)) = p'
	\end{equation}
	for any $p' \in \mathbb Z$ and any switch $\Lambda$. The second equality is a standard Fourier computation.
\end{proof}

	By Prop.\,\ref{def:IEperiodic} we know that $\UE^*(T) [\Lambda_2, \UE(T)]$ is trace class, so that we can apply the previous lemma in direction 2
	\begin{equation}
	\IE  =  \ii \, \Tr_{\ell^2(\mathbb N) \otimes \mathbb C^N} \int_0^{2\pi} \dfrac{\dd k_2}{2\pi}   \Big(\widehat \UE^*(T,k_2) \partial_{k_2}\widehat \UE(T,k_2) \Big) 
	\end{equation}
	Finally, since $\UE(T) = \idE + D(T)$ where both $D$ and $\UE$ are 2-local, their respective Fourier transform $\widehat D(T,k_2)$ and $\widehat \UE(T,k_2)$ are smooth in $k_2$, so that $\partial_{k_2}\widehat \UE(t,k_2) = \partial_{k_2} \widehat D(T,k_2)$ is 1-confined on $\ell^2(\mathbb N) \otimes \mathbb C^N$, namely it is trace class for each $k_2$. Hence trace and integral can be exchanged in the last formula. This proves identity \eqref{IE_winding} of Prop.\,\ref{prop:space_periodic_case}.

\begin{proof}[\textbf{Proof of Prop.\,\ref{prop:defI_per}}]
	By Prop.\,\ref{prop:U_expdecay} $\UB$ is local, then by Lem.\,\ref{lem:Lambdaconfinement} $\UB^*[\Lambda_i,\UB]$ is $i$-confined and $j$-local so that the product of two such terms for $i \neq j$ is trace class according to Lem.\,\ref{lem_icjl_traceclass}, so $\IB$ is well-defined. Then for two switch functions in direction 1, consider their difference $\Delta  \Lambda_1$ and the corresponding difference of indices. We open the inner commutator and separate traces, as we can by Lem.\,\ref{lem_icjl_traceclass}. Then we conjugate by $\UB$ the expression under the first one, so as to obtain
	\begin{equation}\label{comput_IB_invswitch}
	\Delta \IB  = \dfrac{1}{2}\int_0^T \dd t \, \Tr_{\HHB} \Big(  (\partial_t \UB) \UB^* \Big[ \Delta \Lambda_1,[\Lambda_2,\UB]\UB^*\Big]\Big)
	 - \Tr_{\HHB} \Big( \UB^* \partial_t \UB \Big[ \Delta \Lambda_1, \,\UB^*[\Lambda_2,\UB]\Big]\Big).
	\end{equation}
	Each term is vanishing: Since $\Delta \Lambda_1$ is 1-confined, one can open the outer commutator in each one. Up to algebraic manipulation we get for the first one
	\begin{multline}
	\Tr_{\HHB} \Big(  (\partial_t \UB) \UB^* \Big[ \Delta \Lambda_1,[\Lambda_2,\UB]\UB^*\Big]\Big)\\ = \partial_t \Tr_{\HHB} \Big( [\UB, \Lambda_2] \UB^* \Delta \Lambda_1\Big)  
	 + \Tr_{\HHB} \Big( [\UB \Lambda_2 \UB^*,\, (\partial_t \UB) \UB^* \Delta \Lambda_1]\Big)
	\end{multline}
	The first term vanishes when integrated over $t$ since $\UB$ is $t$-periodic. As for the second
	\begin{equation}
	[\UB \Lambda_2 \UB^*,\, (\partial_t \UB) \UB^* \Delta \Lambda_1] = \UB [ \Lambda_2 ,\, \UB^*(\partial_t \UB) \UB^* \Delta \Lambda_1 \UB]\UB^*
	\end{equation}
	Since $\Delta \Lambda_1$ is compactly supported it is 1-confined, so that $\UB^*(\partial_t \UB) \UB^* \Delta \Lambda_1 \UB$ is simultaneously 1-confined and 2-local and the commutator of it with $\Lambda_2$ is trace class according to Lem.\,\ref{lem_icjl_traceclass}, and with vanishing trace according to Lem.\,\ref{lem:com_switch_0}. Thus the first part of \eqref{comput_IB_invswitch} vanishes, and similarly for the second, so that $\IB$ is independent of the choice of $\Lambda_1$. We proceed analogously for $\Lambda_2$.
	
	To show that $\IB$ is an integer, we identify it with a non-commutative odd Chern number \cite{ProdanSchulz16}. Since $\UB$ is periodic in time, consider its inverse Fourier transform along the time direction. Namely for $p,q \in \mathbb Z$ define $\check U_{p,q} = \check U_{0,q-p}$ where
	\begin{equation}
	 \check U_{0,p} = \dfrac{1}{T} \int_0^T \UB(t) \ee^{\ii \frac{2\pi}{T}p t} \dd t
	\end{equation} 
	that acts on $\HHB \otimes \ell^2(\mathbb Z)$. Then consider the following operator appearing in \eqref{defIB_per} up to cyclicity
	\begin{equation}
	O = \Big[\UB^*[\Lambda_1,\UB],\,\UB^*[\Lambda_2,\UB]\Big]\UB^* \quad \Rightarrow \quad  \check O = \Big[\check U^*[\Lambda_1,\check U],\,\check U^*[\Lambda_2,\check U]\Big] \check U^*
	\end{equation}
	since $\UB$ and $\UB^*$ are $t$-periodic and $\Lambda_1$, $\Lambda_2$ naturally extends to $\HHB \otimes \ell^2(\mathbb Z)$. Hence by Lem.\,\ref{lem:periodic} in direction $t$ for $A = O$, $B = \UB$ and $\Lambda_t$ a switch function in direction $t$, we finally get
	\begin{equation}
	\IB = -\ii \pi \Tr_{\HB \otimes \ell^2 (\mathbb Z)} \Big( \check U^* [\Lambda_t, \check U] \Big[\check U^* [\Lambda_1, \check U] ,\, \check U^* [\Lambda_2, \check U]\Big]\Big) =  C_3.
	\end{equation}
	This identifies $\IB$ with $C_3$, the non-commutative version of the odd Chern number in dimension 3, see \cite{ProdanSchulz16}. In particular $\IB \in  \mathbb Z$. Finally the continuity is given by opening the double commutator in expression \eqref{defIB_per} of $\IB$, and noticing that
	\begin{equation}
	\UB^*(t)[\Lambda_1,\UB(t)]\UB^*(t)[\Lambda_2,\UB(t)] = - [\Lambda_1, \UB^*(t)][\Lambda_2,\UB(t)]
	\end{equation}
	is trace class by Lem.\,\ref{lem_icjl_traceclass}, and similarly for the second term where $1 \leftrightarrow 2$. Then consider $\UBo$ and $\UBt$ so that $\UBo(T) = \UBt(T) = \idB$ and denote by $\nu$ one of their common locality exponent. By introducing a mixed term, and inspecting the proof of Lem.\,\ref{lem:Lambdaconfinement}
	\begin{align}
	&\norm{\big([\Lambda_1, \UBo^*(t)][\Lambda_2,\UBo(t)]-[\Lambda_1, \UBt^*(t)][\Lambda_2,\UBt(t)]\big)\ee^{\lambda |n_1|}\ee^{\lambda |n_2|}}\cr
	&\hspace{1cm}\leq B \big(\norm{\UBo}_\nu + \norm{\UBt}_\nu \big) \norm{\UBo-\UBt}_\nu
	\end{align}
	for $\lambda < \nu$, uniformly in time, so that $\UB \mapsto [\Lambda_1,\UB^*][\Lambda_2,\UB]$ is continuous with respect to $\norm{\cdot}_\nu$ and trace norm $\norm{\cdot}_1$. By composition with continuous functions, we deduce that $\UB \mapsto \IB$ is continuous in $\norm{\cdot}_\nu$, and by Prop.\,\ref{cor:U_continuity} that $\HB \mapsto \IB$ is continuous in $\norm{\cdot}_\mu$ as long as $\UB(T) = \idB$.
\end{proof}

Before proving the bulk-edge correspondence, we establish another property of the bulk index that will be used in the general case when $\UB$ is not anymore time-periodic. The proof of this proposition is purely algebraic but quite tedious, we postpone it to App.\,\ref{app:additivity}, so as not to overburden the reading.

\begin{prop}[Additivity of the bulk index]\label{prop:add_IB}
	 Consider $U$ and $V$ two unitary propagators 
	 %associated to two given bulk Hamiltonians and
	  satisfying $U(T) = V(T) = \idB$. Then 
	\begin{equation}
	\IB[UV] = \IB[U] + \IB[V]
	\end{equation}
	where $UV(t) = U(t) V(t)$ on $\HHB$.
\end{prop} 

As we did for the edge index, we can also at the expression of $\IB$ when the bulk Hamiltonian is translation invariant in space. In that case $(\UB)_{\vm,\vn}(t) = (\UB)_{0,\vn-\vm}(t)$ and we define its Fourier transform
\begin{equation}
\widehat \UB(t,k_1,k_2) = \sum_{\vn \in \mathbb Z^2} \ee^{-\ii \mathbf k \cdot \mathbf n} (\UB)_{0,n}(t)
\end{equation}
that defines $\widehat \UB : \mathbb T^3 \rightarrow \mathcal U(N)$, where $\mathbb T^3 = [0,T] \times [0, 2\pi]^2$, namely a matrix valued function periodic in time and quasi-momentum. In analogy with Lem.\,\ref{lem:periodic} we have
\begin{lem}\label{lem:periodic2}
	Let $A$, $B$ and $C$ be three bounded and translation invariant operators on $ \ell^2(\mathbb Z^2)$ with $\mathcal C^2$ Fourier transform denoted by $\widehat A$, $ \widehat B$ and $\widehat C$. Let $\Lambda_1$ and $\Lambda_2$ be two switch functions  in direction 1 and 2. If $A [\Lambda_1, B [\Lambda_2, C]]$ is trace class then
	\begin{equation}
	\Tr_{\ell^2(\mathbb Z^2)} \Big(A [\Lambda_1, B [\Lambda_2, C]]\Big) = (A X_1 B X_2 C)_{00} =  \ii^2 \int \dfrac{\dd^2 k}{(2\pi)^2} \widehat A(k) \partial_{k_1} \big(\widehat  B \, \partial_{k_2}\widehat C\big).
	\end{equation}
\end{lem}
The proof is completely similar to the one of Lem.\,\ref{lem:periodic}, one dimension higher. If we rewrite the operator appearing in the  definition \eqref{defIB_per} of the bulk index as
\begin{equation}\label{aux2}
 \Big[\UB^*[\Lambda_1,\UB],\,\UB^*[\Lambda_2,\UB] \Big] = -\big[\Lambda_1, \UB^*[\Lambda_2, \UB]\big] + \big[\Lambda_2, \UB^*[\Lambda_1, \UB]\big], 
\end{equation}
each term is separately trace class by Lem.\,\ref{lem_icjl_traceclass} since $\UB$ is local. The locality also implies that $\widehat \UB$ is smooth in $k_1$ and $k_2$. We then apply Lem.\,\ref{lem:periodic2} to each part of \eqref{aux2} to end up with identity \eqref{IB_degree} of Prop.\,\ref{prop:space_periodic_case}.

\medskip

Finally, the proof of the bulk-edge correspondence is based on a partial result that improves Prop.\,\ref{prop:defD}. 

\begin{lem}\label{lem:defDelta}
	Let $\HB$ be a bulk Hamiltonian and $\HE,\, \UB, \UE$ the corresponding edge Hamiltonian and bulk and edge propagator. For any switch function in direction 2
	\begin{equation}
	\Delta(t) \equiv [\Lambda_2, \UE(t)] \UE^*(t) - \iota^* [\Lambda_2, \UB(t)] \UB^*(t) \iota
	\end{equation}
	 is trace class on $\HHE$ for every $t \in [0,T]$.
\end{lem}

\begin{proof}
	From Prop.\,\ref{prop:defD} we have that $\UE(t) = \iota^* \UB(t) \iota + D(t)$. We put this expression of $\UE$ in the definition of $\Delta$, use the fact that $\iota^* \Lambda_2 = \Lambda_2 \iota^*$ and $\Lambda_2 \iota = \iota \Lambda_2$ where on the right hand side we mean the extension of $\Lambda_2$ on $\HHB$, and that $\iota\iota^* = P_1$. We end up with
	\begin{equation}
	\Delta = \iota^* [\Lambda_2, \UB] [P_1, \UB^*] \iota + \iota^*[\Lambda_2, \UB ] \iota D^* + [\Lambda_2, D](\iota^* \UB^* \iota + D^*)
	\end{equation}
	where each term is separately trace class by using that $\UB$ is local, $D$ is $1$-confined and $2$-local, $P_1$ is also a switch function and by applying Lem.\,\ref{lem:Lambdaconfinement} and \ref{lem_icjl_traceclass}.
\end{proof}

\begin{proof}[\textbf{Proof of Thm.\,\ref{thm:BE_per}}] We start by the edge index that we rewrite for convenience 
 \begin{equation}
 \IE = \Tr_{\HHE} \Big([\Lambda_2, \UE(T)] \UE^*(T)\Big).
 \end{equation}
 In order to restore a time dependence we introduce a cut-off in direction 1. For $r \in \mathbb N$ take $Q_{1,r} = \chi_{n_1 < r}$ on $\HHE$ and note that $Q_{1,r} = \idE - P_{1,r}$ where $P_{1,r}$ is a also a switch function. Since $\UB(T) = \idB$ the operator in the previous expression of $\IE$ is nothing but $\Delta(T)$  which is trace class. Moreover $Q_{1,r} \rightarrow \idE$ strongly when $r \rightarrow \infty$, so that
 \begin{equation}
 \IE = \lim_{r\rightarrow \infty} \IE^r \qquad \IE^r = \Tr_{\HHE} \Big([\Lambda_2, \UE(T)] \UE^*(T) Q_{1,r} \Big)
 \end{equation} 
 Then we can rewrite
 \begin{equation}
 \IE^r = \int_0^T \dd t\, \partial_t \Tr_{\HHE} \Big([\Lambda_2, \UE(t)] \UE^*(t) Q_{1,r} \Big) \equiv \int_0^T \dd t\, I_r(t)
 \end{equation} 
 Indeed
 \begin{equation}
 I_r(t) = \Tr_{\mathcal \HHE} \Big([\Lambda_2, \partial_t \UE(t)] \UE^*(t) Q_{1,r} \Big)
  -  \Tr_{\mathcal \HHE} \Big([\Lambda_2, \UE(t)] \UE(t)^* (-\ii \HE)(t) Q_{1,r} \Big) 
 \end{equation}
 is trace class by Lem.\,\ref{lem:Lambdaconfinement} and \ref{lem_icjl_traceclass} since $\UE$, $\HE$ and $\partial_t \UE = -\ii \HE \UE$ are local and $Q_{1,r}$ is trivially 1-confined and 2-local simultaneously. From now on we drop the time dependence in $I_r(t)$. By Lem.\,\ref{lem:com_switch_0} we have
 \begin{align}
0 &= \Tr_{\mathcal \HHE} \Big(\big[\Lambda_2, (\partial_t \UE) \UE^* Q_{1,r} \big]\Big)\nonumber\\
 &= \Tr_{\mathcal \HHE} \Big([\Lambda_2, \partial_t \UE] \UE^* Q_{1,r} \Big) + \Tr_{\mathcal \HHE} \Big(\partial_t \UE[\Lambda_2, \UE^* Q_{1,r}]  \Big)
 \end{align}
 where on the r.h.s the first term appears in $I_r(t)$ and the second one can be expanded by using that $[\Lambda_2, Q_{1,r}] = 0$, $[\Lambda_2, \UE^*] = - \UE^*[\Lambda_2, \UE] \UE$ and $(\partial_t \UE)\UE^* = -\ii \HE$. We end up with
 \begin{equation}\label{comput_BE_edgeIrt}
 I_r(t) = \Tr_{\HHE} \Big([\Lambda_2, \UE] \UE^* \, [P_{1,r}\,,\ii \HE ]\Big)
 \end{equation}
 where we have also used that $[\HE,Q_{1,r}] =[P_{1,r},\HE]$. This expression can now be  recast as a bulk expression. By Lem.\,\ref{lem:defDelta}, Def.\,\ref{def:edge_Hamiltonian} of $\HE$, and by denoting $\iota^* P_{1,r} = P_{1,r} \iota^*$ where on the right hand side $P_{1,r} = \chi_{n_1\geq r}$ on $\HHB$, and similarly with $\iota$, we get
 \begin{equation}
 I_r(t)
  = \Tr_{\HHE} \Big(\iota^* [\Lambda_2, \UB] \UB^* \, [P_{1,r}\,, \ii \HB] \iota + \iota^* \big[[\Lambda_2, \UB] \UB^*,P_1 \big]
   [P_{1,r}\,, \ii H_B] \iota+  \Delta [P_{1,r}\,, \ii \HE] \Big)
 \end{equation}
 where we have used that $\iota \iota^* = P_1$ and $\iota^*\iota=\idE$. The traces of the last two terms vanish in the limit $r \rightarrow \infty$. Indeed both $[[\Lambda_2, \UB] \UB^*,P_1 \big]$ and $\Delta$ are trace class according to Lem.\,\ref{lem_icjl_traceclass} and \ref{lem:defDelta} respectively, and $P_{1,r} \rightarrow 0$ strongly, so that 
 \begin{equation}
 \Tr_{\HHB} \Big( \iota^* \big[[\Lambda_2, \UB] \UB^*,P_1 \big] [P_{1,r}\,, \ii H_B ] \iota+  \Delta [P_{1,r}\,, \ii \HE ] \Big) \mathop{\longrightarrow}_{r \rightarrow \infty} 0.
 \end{equation}
 Finally note that for any trace class operator $O$ on $\HHE$ one has $\Tr_{\HHE}(O) = \Tr_{\HHB}(\iota O \iota^*)$, so that
 \begin{equation}
 I_r(t) = \Tr_{\HHB}\Big( [\Lambda_2, \UB] \UB^* [P_{1,r}\,, \ii \HB ] P_1  \Big) + o(1)
  \end{equation}
  where we have used again that $\iota \iota^* = P_1$, $P_1^2 = P_1$ and the cyclicity of trace. 
  
  The next step is to show that $P_1$ can be omitted in the previous expression. Intuitively, $[P_{1,r}\,, \HB ]$ is confined along $n_1 = r$ so that its contribution for $n_1 <0$ vanishes exponentially when $r$ is big enough. More explicitly we compute
  \begin{equation}
\Tr_{\HHB}\Big( [\Lambda_2, \UB] \UB^* [P_{1,r}\,, \ii \HB ] (1 - P_1)  \Big) 
= -\Tr_{\HHB}\Big( P_{1,r} [(1 - P_1),\ii \HB] [\Lambda_2, \UB] \UB^*   \Big) \rightarrow 0
  \end{equation}
  where we  used that $P_{1,r}(1-P_1)=0$ and the cyclicity of the trace. Since  $[(1 - P_1), \ii \HB ] [\Lambda_2, \UB] \UB^* $ is trace class and $P_{1,r} \rightarrow 0$ strongly then the previous expression vanishes in the limit $r \rightarrow \infty$. Moreover the trace on the l.h.s. can be split into two traces so that
   \begin{equation}
   I_r(t) = \Tr_{\HHB}\Big( [\Lambda_2, \UB] \UB^* [P_{1,r}\,,\ii \HB ] \Big) + o(1)
   \end{equation}
   as announced. This expression is nothing but \eqref{comput_BE_edgeIrt} where we have replaced every E by B and up to corrections vanishing in the limit $r \rightarrow \infty$, even when integrated over the compact interval $[0,T]$.
   
   The final step is to get back the expression of the bulk index as in \eqref{defIB_per}. First we rewrite $\ii \HB = -(\partial_t \UB) \UB^*$. Then we have the following identity
   \begin{align}
   \label{compute_BE_algebraicI}
   &\Tr_{\HHB}\Big( [\Lambda_2, \UB] \UB^* [(\partial_t \UB) \UB^*, P_{1,r} ] \Big) \cr&= \dfrac{1}{2} \Tr_{\mathcal H_B} \Big( (\partial_t \UB) \UB^* \Big[[P_{1,r},\, \UB]\UB^*, \,[\Lambda_2, \UB] \UB^* \Big]\Big)\cr
  &\quad  + \dfrac{1}{2}\partial_t \Tr_{\HHB} \Big(\Big[[\Lambda_2, \UB]  , P_{1,r}  \Big]  \UB^*\Big).
   \end{align}
   This identity is purely algebraic but quite tedious to show so we postpone the computation to App.\,\ref{app:BE_algebraicI}. Since $\UB$ is periodic, the second term vanishes when integrated over time. Conjugating the first one by $\UB^*$ and $\UB$ and putting all together we get
   \begin{equation}
   \IE = \lim_{r \rightarrow \infty} \dfrac{1}{2}\int_0^T \dd t \, \Tr_{\mathcal H_B} \Big( \UB^*\partial_t \UB  \Big[\UB^*[P_{1,r}\,, \UB], \,\UB^*[\Lambda_2, \UB]  \Big]\Big)
   \end{equation}
   but on the right hand side we recognize the bulk index expression, that is independent of the choice of switch function. In particular $P_{1,r}$ can be replaced by $P_1$ or any $\Lambda_1$, so that the limit is trivial and we get $\IE = \IB$.
\end{proof}

\subsection{General case}

In the general case the bulk-edge correspondence is a corollary of Thm.\,\ref{thm:BE_per}, so we only need to check that this theorem applies, namely that the effective Hamiltonian $\HB^\varepsilon$ from Def.\,\ref{def:Heff} has the required properties, in particular that it is local. By spectral decomposition

\begin{equation}\label{Hepsilon_spec}
\UB(T) = \int_{\mathcal S^1} \lambda \dd P(\lambda) \quad \Rightarrow \quad   \HB^{\varepsilon} = \dfrac{\ii}{T} \int_{\mathcal S^1} \log_{-T \varepsilon}( \lambda) \dd P(\lambda)
\end{equation}
where the integration is done over the unit circle and $\dd P(\lambda)$ is the spectral measure of $U_B(T)$.

\begin{prop}\label{prop:Heff_local}
	Let  $\HB^\varepsilon$ be an effective Hamiltonian constructed from a bulk Hamiltonian $\HB$. Then $\HB^\varepsilon$ is local, namely it exists $\lambda^* >0$ such that for $0 \leq \lambda < \lambda^*$ 
	\begin{equation}
	\Vert \ee^{-\lambda f} \HB^\varepsilon \ee^{\lambda f} - H_B^\varepsilon \Vert \leq \beta_\lambda
	\end{equation}
	 for any Lipschitz function $f$. Moreover $\beta_\lambda \rightarrow 0$ for $\lambda \rightarrow 0$.
\end{prop}

	\begin{proof}
		As in the proof of Prop.\,\ref{prop:U_expdecay} we consider bounded $f_n$ instead of $f$ to work with bounded operators, get a uniform estimate independent of $n$ allowing to consider the $n \rightarrow \infty$ limit. We compute $\HB^\varepsilon$ through the resolvent formula
		\begin{equation}
		\HB^\varepsilon = - \dfrac{1}{2\pi \ii} \dfrac{\ii}{T} \int_\Gamma \dd z  \log_{-T \varepsilon}(z) R_{\UB}(z) 
		\end{equation}
		where $R_{\UB}(z) \equiv (\UB(T) - z)^{-1}$ and $\Gamma$ is illustrated in Fig.\,\ref{fig:gammalog}. In particular one has
		\begin{equation}
		\ee^{-\lambda f_n} \HB^\varepsilon \ee^{\lambda f_n} = - \dfrac{1}{2\pi \ii} \dfrac{\ii}{T} \int_\Gamma \dd z  \log_{-T \varepsilon}(z) R_{U_n}(z), 
		\end{equation}
		where we defined $U_n(T) \equiv \ee^{-\lambda f_n} \UB(T) \ee^{\lambda f_n}$.
		\begin{figure}[htb]
			\centering
			\begin{tikzpicture}[scale=0.8]
			\newdimen\r
			\pgfmathsetlength\r{1.5cm}
			\draw[thick] (0,0) circle (\r);
			\draw[line width=0.15cm,DeepSkyBlue3] (20:\r) arc (20:130:\r);
			\draw[line width=0.15cm,DeepSkyBlue3] (-40:\r) arc 
			(-40:-120:\r);
			
			\draw[densely dashed] (190:0) -- (190:{1.2*\r}) node[anchor=north east] 
			{$\,\ee^{-\ii T\varepsilon}$};
			\draw (-0.05,-0.05) -- (0.05,0.05);
			\draw (0.05,-0.05) -- (-0.05,0.05);
			\draw (0,0) node[below] {0};
			
			\draw[red,thick] (140:0.8*\r) arc (140:-130:0.8*\r);
			\draw[red,thick] (140:1.2*\r) arc (140:-130:1.2*\r);
			
			\draw[red,thick] (140:1.2*\r) -- (140:0.8*\r);
			\draw[red,thick] (-130:1.2*\r) -- (-130:0.8*\r);
			
			\draw[red] (2.2,0) node {$\Gamma$};
			
			\draw[red,thick] (1.2*\r,0) -- (1.1*\r,-0.2);
			\draw[red,thick] (1.2*\r,0) -- (1.3*\r,-0.2);			
			\end{tikzpicture}
			\caption{\label{fig:gammalog}Contour $\Gamma$ to compute the logarithm with branch cut in the spectral gap of $\UB(T)$.}
		\end{figure}
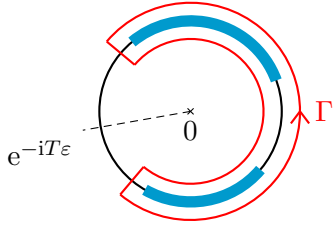
		
		The usual resolvent identity leads to
		\begin{equation}
		R_{U_n}(z) \Big( 1 + (U_n(T) - \UB(T))) R_{\UB}(z) \Big) = R_{\UB}(z).
		\end{equation}
		First we have
		\begin{equation}
		\Vert R_{\UB}(z) \rVert \leq \dfrac{1}{\mathrm{dist}(z,\sigma(U(T)))} \leq \dfrac{1}{\eta}
		\end{equation}
		where $\eta = \mathop{\inf}_{z \in  \Gamma} \big(\mathrm{dist}(z,\sigma(U(T))) \big) >0 $. Then from Prop.\,\ref{prop:U_expdecay}, we know that for $\lambda$ sufficiently small, let say $\lambda < \lambda^*$ one has
		\begin{equation}
		\Vert U_n(T) - \UB(T) \Vert \leq \alpha_\lambda < \eta
		\end{equation}
		independently from $n$. This implies that $1+(U_n(T)-\UB(T)) R_{\UB}(z)$ is invertible for $z \in \Gamma$. Thus
		\begin{equation}
		\Vert R_{U_n}(z) \Vert = \left\lVert R_{\UB}(z)\Big(1+(U_n(T)-\UB(T)) R_{\UB}(z)\Big)^{-1} \right\rVert \leq \dfrac{1}{\eta-\alpha_\lambda},
		\end{equation}
		so that $R_{U_n}(z)$ is bounded independently from $n$. We compute
		\begin{equation}
		\ee^{-\lambda f_n} \HB^\varepsilon \ee^{\lambda f_n} - \HB^\varepsilon = - \dfrac{1}{2\pi \ii} \dfrac{\ii}{T} \int_\Gamma \dd z  \log_{-T \varepsilon}(z) (R_{U_n}(z)-R_{\UB}(z)) 
		\end{equation}
		then again by the resolvent identity and the previous estimates
		\begin{align}
		\left\lVert R_{U_n}(z)-R_{\UB}(z) \right\rVert & = \left\lVert R_{U_n}(z) (\UB(T)-U_n(T)) R_{\UB}(z) \right\rVert
		\leq \dfrac{1}{\eta-\alpha_\lambda} \alpha_\lambda \dfrac{1}{\eta}.
		\end{align}
		Finally
		\begin{equation}
		\left\lVert \ee^{-\lambda f_n} \HB^\varepsilon \ee^{\lambda f_n} - \HB^\varepsilon \right\rVert \leq \dfrac{|\Gamma|}{2\pi T} \sup_{z \in \Gamma} \{|\log(z)| \} \, \dfrac{\alpha_\lambda}{\eta(\eta-\alpha_\lambda)} \quad  \equiv \beta_\lambda
		\end{equation}
		for $\lambda < \lambda^*$ such that $\alpha_\lambda < \eta$. The term on the r.h.s. is finite, independent of $n$ and goes to 0 when $\lambda \rightarrow 0$. Thus we have the same when $n \rightarrow \infty$, leading to the result.
	\end{proof}

 We then study the influence on the choice of $\varepsilon$, first on $\HB^\varepsilon$ then on the bulk index. The proof of Lem.\,\ref{lem:Heff_infl_epsilon} is straightforward. Both identities come from the spectral decomposition \eqref{Hepsilon_spec}, and the properties of the logarithm. The first one from the fact that $\log_{\alpha + 2\pi } = \log_{\alpha} +  2\pi\ii $ and the second from 
	\begin{equation}
	\log_{\alpha'}(\ee^{\ii \phi})- \log_{\alpha}(\ee^{\ii \phi}) = \left\lbrace\begin{array}{ll}
	0, & (0 \leq \phi < \alpha) \\
	2 \pi \ii,   & (\alpha < \phi < \alpha') \\
	0,  & (\alpha' < \phi < 2\pi) \\
	\end{array}\right..
	\end{equation} 

As we shall see, \eqref{Heff2pi} tells us that we can restrict $\varepsilon$ to any interval of length $2\pi/T$ and \eqref{comp_Heff} compares two effective Hamiltonians in that interval. In particular they coincide when $\ee^{-\ii T \varepsilon}$ and $\ee^{-\ii T \varepsilon'}$ belong to the same gap. 

\begin{proof}[\textbf{Proof of Prop.\,\ref{prop:influence_varep}}]
	By construction, $\HB^\varepsilon$ is time independent, so that the relative evolution is
	\begin{equation}
	U_\mathrm{B,rel}^\varepsilon(t) = \left\lbrace \begin{array}{ll}
	\UB(2t), & (0 \leq t \leq T/2) \\
	\exp\big(-\ii 2(T-t) \HB^\varepsilon\big), &  (T/2 \leq t \leq  T)
	\end{array}\right..
	\end{equation}
	From \eqref{Heff2pi} we deduce $U_\mathrm{B,rel}^{\varepsilon+2\pi/T} = U_\mathrm{B,rel}^\varepsilon U_{\Id}$ where 
	    \begin{equation}\label{defUIn}
	    U_{\Id}(t) = 	\left\lbrace \begin{array}{ll}
	    \idB, &  (0 \leq t \leq T/2) \\
	    \exp\Big(-\ii 2(T-t)  \dfrac{2\pi }{T} \idB \Big),   &  (T/2 \leq t \leq  T)
	    \end{array}\right.
	    \end{equation}
	that also satisfy $U_{\Id}(T)=\idB$. Moreover $\IB[U_{\Id}]= 0$ since $U_{\Id}$  acts trivially on $\HHB$, so that by the additivity from Prop.\,\ref{prop:add_IB} we deduce $\IB(\varepsilon + 2\pi /T) = \IB(\varepsilon)$.
	
	Similarly, for $0 \leq  \varepsilon' - \varepsilon< 2 \pi$ we get from \eqref{comp_Heff} that $U_\mathrm{B,rel}^{\varepsilon'}\! = U_\mathrm{B,rel}^\varepsilon U_{P_{\varepsilon,\varepsilon'}}$ where $U_{P_{\varepsilon,\varepsilon'}}$ is similar to $U_{\Id}$ but with $P_{\varepsilon,\varepsilon'}$ instead of $\idB$ in \eqref{defUIn}. It is then shown in App.\,\ref{app:chern} that
	\begin{equation} \label{id_chern}
	\IB[U_{P,\varepsilon,\varepsilon'}] =-  2\pi \ii \,\Tr\Big( P_{\varepsilon,\varepsilon'} \Big[  \big[\Lambda_1,P_{\varepsilon,\varepsilon'} \big],  \big[\Lambda_2, P_{\varepsilon,\varepsilon'} \big] \Big] P_{\varepsilon,\varepsilon'} \Big) 
	= c(P_{\varepsilon,\varepsilon'})  \in \mathbb Z
	\end{equation}
	which is the Kubo-St\v{r}eda formula or non-commutative Chern number of $P_{\varepsilon,\varepsilon'}$ from the Quantum Hall effect \cite{AvronSeilerSimon94}. We conclude by the additivity property of $\IB$ from Prop.\,\ref{prop:add_IB}.
\end{proof}

We finally deal with continuity properties.

\begin{prop}\label{HeffcontinousinUB}
	It exists $\lambda,\,\nu >0$ such that $\UB \mapsto \HB^\varepsilon$ is continuous with respect to $\norm{\cdot}_\nu$ and $\norm{\cdot}_\lambda$ as long as $\ee^{-\ii T \varepsilon}$ lies in a spectral gap of $\UB(T)$.
\end{prop}
\begin{proof}
	Let $\UBo(T)$ and $\UBt(T)$ with $\ee^{-\ii T \varepsilon}$ belonging to a common spectral gap. Take $0 < \lambda < \lambda^*$ from Prop.\,\ref{prop:Heff_local} so that $\HBo^\varepsilon$ and $\HBt^\varepsilon$ are both local with common exponent $\lambda$. Similarly to the proof of Prop.\,\ref{prop:Heff_local}
	\begin{equation}
	\ee^{-\lambda f} (\HBo^\varepsilon-\HBt^\varepsilon) \ee^{\lambda f}
	 = - \dfrac{1}{2\pi \ii} \dfrac{\ii}{T} \int_\Gamma \dd z  \log_{-T \varepsilon}(z) R_{U_{1f}}(z) \ee^{-\lambda f} (\UBo-\UBt)(T) \ee^{\lambda f} R_{U_{2f}}(z) 
	\end{equation}
	where $U_{if} = \ee^{-\lambda f}U_{\mathrm B,i}(T) \ee^{\lambda f}$, $\Gamma$ is a contour common to $\UBo(T)$ and $\UBt(T)$, and where we have used the resolvent identity. We know from the previous proof that $R_{U_{1f}}$ and $R_{U_{2f}}$ are both bounded for $z \in \Gamma$. By Prop.\,\ref{prop:U_expdecay} and Cor.\,\ref{cor:localbis}, we know that $\ee^{-\lambda f} (\UBo-\UBt)(T) \ee^{\lambda f}$ is bounded by $\norm{(\UBo-\UBt)}_\nu$ for some $\nu > \lambda$. Thus
	\begin{equation}
	\norm{\ee^{-\lambda f} (\HBo^\varepsilon-\HBt^\varepsilon) \ee^{\lambda f}} \leq B \norm{(\UBo-\UBt)}_\nu
	\end{equation}
	and consequently we have a similar estimate for $\Vert (\HBo^\varepsilon-\HBt^\varepsilon)\Vert_\lambda$.
\end{proof}

Together with Prop.\,\ref{cor:U_continuity}, we deduce that $\HB \mapsto U_{\mathrm{B,rel}}^\varepsilon$ is continuous, respectively with $\norm{\cdot}_\mu$ and $\norm{\cdot}_\lambda$. This proves Cor.\,\ref{cor:IB_contiuity} on homotopy invariance of $\IB$.

\subsection{Interface index properties}

Note that one can also embed the edge Hamiltonians instead of gluing the bulk ones. Namely by considering $\mathbb N_- = \mathbb Z \setminus \mathbb N$,\, $\HHE^- = \ell^2(\mathbb N_- \times \mathbb Z) \otimes \mathbb C^N$, $\iota_- : \HHE^- \rightarrow \HHB$ and $\iota_-^*: \HHB \rightarrow \HHE^-$, one has
	\begin{equation}\label{HIedges}
	\HI = \iota H_{\mathrm{E},1} \iota^* + \iota_- H_{\mathrm{E},2}^- \iota_-^* + H_\mathrm{int}
	\end{equation}
where we have defined $\HE^- = \iota_-^* \HB \iota$, namely the edge Hamiltonian on the other half space, and used \eqref{J_partialisom} and similarly $\iota_-^* \iota_- = \Id_{\HHE^-}$ and $\iota_- \iota_-^* = 1-P_1$. 

\begin{lem}\label{lem:defDI}
	Let $\HI$ be the interface Hamiltonian from Def.\,\ref{def:interfaceH}. Then the corresponding propagator satisfies
	\begin{equation}
	\UI(t) = \iota U_{\mathrm{E},1}(t) \iota^* +  \iota_- U_{\mathrm{E},2}^-(t) \iota_-^* + \DI(t)
	\end{equation}
	where $U_{\mathrm{E},2}^-$ is generated by $\HE^-$ and $\DI(t)$ is simultaneously 1-confined and 2-local.
\end{lem}

\begin{proof}
 From \eqref{HIedges} and Def.\,\ref{def:interfaceH}
 \begin{equation}
 \iota^* H_I \iota = H_{\mathrm E,1} + \iota^* H_\mathrm{int} \iota,
 \end{equation}
 where $  \iota^* H_\mathrm{int} \iota$ being simultaneously 1-confined and 2-local on $\HHE$ plays the role of a boundary condition as in Prop.\,\ref{prop:general_BC}, and $\HI$ that of the bulk Hamiltonian for $H_{\mathrm E,1}$.  In particular by adapting the proof of Prop.\,\ref{prop:defD}
 \begin{equation}
U_{\mathrm E,1}(t) = \iota^* \UI(t) \iota + \DIo(t) \quad \Rightarrow \quad P_1 U_I P_1 = \iota ( U_{\mathrm E,1} - \DIo) \iota^*
 \end{equation}
 where $\DIo$ is simultaneously 1-confined and 2-local and where we have used \eqref{J_partialisom}. Similarly:
 \begin{equation}
(1-P_1) U_I (1-P_1) = \iota_- ( U_{\mathrm E,2}^- - \DIt) \iota_-^*
 \end{equation}  on the other half space. By decomposing $\UI$ over subspaces associated to $P_1$ and $1-P_1$ we get
 \begin{equation}
 \UI = \iota U_{\mathrm E,1} \iota^*  + \iota_- U_{\mathrm E,2}^- \iota_-^* - \iota \DIo \iota^* - \iota_- \DIt \iota_-^* + P_1 \UI (1-P_1) + (1-P_1) \UI P_1
 \end{equation}
Each of the last four terms is simultaneously 1-confined and 2-local from the properties of $\DIo$, $\DIt$ and the fact that $\UI$ is local. Together they define $\DI$.
\end{proof}

\begin{proof}[\textbf{Proof of Prop.\,\ref{prop:interfaceindex}}]
	From Lem.\,\ref{lem:defDI} we get $\UI(T)$ in terms of $U_{\mathrm{E},1}(T)$ and $U_{\mathrm{E},2}(T)$, but since the corresponding bulk propagator are not $\idB$, we need to normalize $\UI(T)$ as in \eqref{defII}. In particular consider the special interface 
	\begin{equation}
	\HBt(t)= \iota H_{\mathrm{E},2} \iota^* + \iota_- H_{\mathrm{E},2}^- \iota_-^* + P_1 \HBt (1-P_1) + (1-P_1) \HBt P_1  
	\end{equation}
	which is nothing but an interface decomposition of $\HBt$. Lem.\,\ref{lem:defDI} gives
	\begin{equation}
	\UBt(T)  = \iota U_{\mathrm{E},2}(T) \iota^* +  \iota_- U_{\mathrm{E},2}^-(T) \iota_-^* + \widetilde \DI(T)
	\end{equation}
	Hence by Lem.\,\ref{lem:defDI} applied for $\UI$ and $\UBt$ we deduce after some algebra
	\begin{equation}\label{comput_Dhat}
	\UBt^*\UI(T) = \iota U_{\mathrm{E},2}^* U_{\mathrm{E},1} (T) \iota^* + \iota_- \Id_{\HHE^-} \iota_-^*  + \widehat D(T)
 	\end{equation}
 	where we have used $\iota_-^* \iota = 0$ and $\iota^* \iota_- =0$ and where $\widehat D(T)$ is simultaneously 1-confined and 2-local. Finally, from Prop.\,\ref{prop:defD} and the fact that $\UBo(T) = \UBt(T)$ we deduce at $t=T$
 	\begin{equation}
 	U_{\mathrm{E},2}^* U_{\mathrm{E},1} = \idE + \iota^* \UBt^* [P_1, \UBo] \iota + \iota^*  \UBt^* \iota D_1 + D_2^* (\iota^* \UBo \iota + D_1)
 	\end{equation}
 	where each term except $\idE$ is simultaneously 1-confined and 2-local. Putting all together, we deduce that $[\Lambda_2, \UBt^*\UI(T)]$ is trace class so that $\II$ is finite. Similarly to the proof of Prop.\,\ref{defIB_per}, $\II$ can be identified with an index of a pair of projections so it is integer valued, independent of $\Lambda_2$ and continuous (with the local norm) in $\UBt$ and $\UI$. In particular consider the deformation of the previous derivation to the sharp interface where the two halves are disconnected
 	\begin{equation}
 	\HI = \iota H_{\mathrm{E},1} \iota^* + \iota_- H_{\mathrm{E},2}^- \iota_-^*,\qquad \widetilde{H}_{\mathrm B,2} = \iota H_{\mathrm{E},2} \iota^* + \iota_- H_{\mathrm{E},2}^- \iota_-^*
 	\end{equation}
 	In that case the corresponding evolutions are also disconnected so that in \eqref{comput_Dhat} $\widehat D(T) = 0$, and we deduce $\II = \IE^\mathrm{rel}$ from expression \eqref{IErel} and $\Tr_{\HHB}(\iota O \iota^* ) = \Tr_{\HHE}(O)$.
\end{proof}

\appendix

\section{Some algebraic computations \label{sec:app}}
\subsection{Additivity of the bulk index \label{app:additivity}}

Here we prove Prop.\,\ref{prop:add_IB}. It is purely algebraic but quite tedious. From the definition \eqref{defIB_per} of $\IB$ we compute 

\begin{align}\label{dem_additivity1}
&\Tr \,(UV)^* \partial_t (UV) \Big[ (UV)^* \big[\Lambda_1 , UV\big], (UV)^* \big[\Lambda_2 , UV\big]\Big] \cr
&= \Tr\, U^* \partial_t U \Big[ U^*\big[\Lambda_1 , U\big], U^* \big[\Lambda_2 , U\big] \Big] + \Tr\, V^* \partial_t V \Big[ V^*\big[\Lambda_1 , V\big], V^* \big[\Lambda_2 , V\big] \Big] \cr
&\hspace{0.5cm}\left.\begin{array}{l}
+ \Tr\, U^* \partial_t U \Big[\big[\Lambda_1 , V\big] V^*, U^* \big[\Lambda_2 , U\big] \Big] + \Tr\, U^* \partial_t U \Big[U^*\big[\Lambda_1 , U\big] ,  \big[\Lambda_2 , V\big] V^*\Big] \\
+ \Tr\, U^* \partial_t U \Big[\big[\Lambda_1 , V\big] V^*,  \big[\Lambda_2 , V\big]V^* \Big] +  \Tr\,  (\partial_t V) V^* \Big[U^* \big[\Lambda_1 , U\big] , U^* \big[\Lambda_2 , U\big] \Big] \\
+ \Tr\,  (\partial_t V) V^* \Big[ \big[\Lambda_1 , V\big]V^* , U^* \big[\Lambda_2 , U\big] \Big] + \Tr\,  (\partial_t V) V^* \Big[ U^* \big[\Lambda_1 , U\big] ,  \big[\Lambda_2 , V\big] V^*\Big] 
\end{array}\right\rbrace \equiv R \cr
\end{align}

where we have used Leibniz rule for $\partial_t$ and $[\Lambda_i, \cdot\,]$ and the cyclicity of trace (note that each written term is trace class by Lem.\,\ref{lem_icjl_traceclass} as long as $U$ and $V$ are local). The two first terms in the latter equation corresponds to the index of $U$ and $V$ when integrated over time. After a bit of algebra one can check that the remaining last three lines are actually equal to
\begin{align}
R = & - \Tr\,\Big[\Lambda_1, U^* \big[\Lambda_2,U\big] (\partial_t V) V^* - U^* \partial_t U \big[ \Lambda_2, V\big] V^*\Big] \cr
& - \Tr\,\Big[\Lambda_2, U^* \partial_t U \big[ \Lambda_2, V\big] V^*- U^* \big[\Lambda_1,U\big] (\partial_t V) V^* \Big] \cr
& - \Tr\,\partial_t\Big( U^* \big[\Lambda_1, U\big] \big[\Lambda_2, V\big] V^* - U^* \big[\Lambda_2, U\big] \big[\Lambda_1, V\big] V^*\Big).
\end{align}
The first two terms are trace class with vanishing trace according to Lem.\,\ref{lem:com_switch_0} and the last one is a total time derivative that vanishes when integrated over time since $U$ and $V$ are periodic by assumption. Thus $R$ vanishes when integrated from $0$ to $T$
so that \eqref{dem_additivity1} leads to the expected result.  Note that this proof is nothing but the one given in \cite{Lyon15} in the periodic case adapted to the the derivatives $[\Lambda_i, \cdot\,]$ for the space directions.

\subsection{Proof of identity \eqref{compute_BE_algebraicI} \label{app:BE_algebraicI}}

In the following all the traces involved are finite using that $\UB$ is local and Lem.\,\ref{lem:Lambdaconfinement} and \ref{lem_icjl_traceclass}. On the one hand we can expand
\begin{align}\label{BE_app_onehand}
& \Tr_{\HHB} \Big( [\Lambda_2, \UB] \UB^* \, [(\partial_t \UB)\UB^*, P_{1,r}] \Big)
 = \Tr_{\HHB} \Big( [\Lambda_2, \UB] \UB^* \, [\partial_t \UB, P_{1,r}]\UB^* \Big)\nonumber\\
&\qquad -  \Tr_{\HHB} \Big( [\Lambda_2, \UB] \UB^*\,\partial_t \UB \UB^* \,[\UB, P_{1,r}]\UB^* \Big)
\end{align}
and on the other hand we notice that, due to Lem.\,\ref{lem:com_switch_0} 
\begin{align}
0 &=\Tr_{\HHB}  \Big[[\Lambda_2, \UB] \UB^*\, (\partial_t \UB)\UB^*, P_{1,r}  \Big] 
= \Tr_{\HHB} \Big( [\Lambda_2, \UB] \UB^* \, [(\partial_t \UB)\UB^*, P_{1,r}] \Big)\nonumber\\
&\qquad + \Tr_{\HHB} \Big(\Big[[\Lambda_2, \UB] \UB^*\, , P_{1,r}  \Big] (\partial_t \UB)\UB^*\Big).
\end{align}
The first term is the one of interest and the second can be expanded 
\begin{align}
&\Tr_{\HHB} \Big(\Big[[\Lambda_2, \UB] \UB^*\, , P_{1,r}  \Big] (\partial_t \UB)\UB^*\Big) 
= \Tr_{\HHB} \Big(\Big[[\Lambda_2, \UB]  , P_{1,r}  \Big] \UB^*\,(\partial_t \UB)\UB^*\Big)\nonumber\\
&\qquad - \Tr_{\HHB} \Big( [\Lambda_2, \UB] \UB^*\,[\UB, P_{1,r}]\UB^*\,\partial_t \UB \UB^*  \Big).
\end{align}
Then we rewrite the first term appearing here using an integration by parts, namely
\begin{align}
&\Tr_{\HHB} \Big(\Big[[\Lambda_2, \UB]  , P_{1,r}  \Big] \UB^*\,(\partial_t \UB)\UB^*\Big) \cr
&= -\partial_t \Tr_{\HHB} \Big(\Big[[\Lambda_2, \UB]  , P_{1,r}  \Big]  \UB^*\Big) + \Tr_{\HHB} \Big(\Big[[\Lambda_2, \partial_t \UB]  , P_{1,r}  \Big]  \UB^*\Big).
\end{align}
Finally, similarly as before,
\begin{align}
 & \Tr_{\HHB} \Big(\Big[[\Lambda_2, \partial_t \UB]  , P_{1,r}  \Big]  \UB^*\Big) \cr
   & = - \Tr_{\HHB} \Big(\Big[[\partial_t \UB, P_{1,r}], \Lambda_2\Big] \UB^*\Big) \cr &
 =- \Tr_{\HHB} \Big(\Big[[ \partial_t \UB , P_{1,r} ]\UB^* , \Lambda_2 \Big]\Big) +  \Tr_{\HHB} \Big([ \partial_t \UB , P_{1,r} ][\UB^* , \Lambda_2 ]\Big) \
\end{align}
where the first term vanishes by Lem.\,\ref{lem:com_switch_0}. Putting together the last three equations, we deduce
\begin{align}\label{BE_app_otherhand}
\Tr_{\HHB} &\Big( [\Lambda_2, \UB] \UB^* \, [(\partial_t \UB)\UB^*, P_{1,r}] \Big)\cr
& =  \partial_t \Tr_{\HHB} \Big(\Big[[\Lambda_2, \UB]  , P_{1,r}  \Big]  \UB^*\Big) -  \Tr_{\HHB} \Big([\UB^* , \Lambda_2 ][ \partial_t \UB , P_{1,r} ] \Big) \cr
& \hspace{0.5cm} + \Tr_{\HHB} \Big( [\Lambda_2, \UB] \UB^*\,[\UB, P_{1,r}]\UB^*\,\partial_t \UB \UB^*  \Big).
\end{align}
 Noticing that $[\UB^*, \Lambda_2] = - \UB^* [\UB,\Lambda_2] \UB^*$ and summing \eqref{BE_app_onehand} and \eqref{BE_app_otherhand} we get identity \eqref{compute_BE_algebraicI}.
 
\subsection{Proof of identity \eqref{id_chern}\label{app:chern}} 

We first rewrite $P = P_{\varepsilon,\varepsilon'}$ and $\Tr = \Tr_{\HHB}$. Note that $P$ is a spectral projector of $\UB(T)$ so it is also local and all the following traces are finite. By definition the first half of the time integral is trivial for $U_{P,\varepsilon,\varepsilon'}$ (defined similarly to \eqref{defUIn}). So that up to a change of variables 
\begin{equation}
\IB[U_{P,\varepsilon,\varepsilon'}] = \dfrac{\ii \pi}{T} \int_0^T \dd t \Tr\Big( P \Big[ \ee^{2\pi \ii \frac{t}{T} P} \big[\Lambda_1, \ee^{-2\pi \ii \frac{t}{T} P} \big], \ee^{2\pi \ii \frac{t}{T} P} \big[\Lambda_2, \ee^{-2\pi \ii \frac{t}{T} P} \big] \Big]\Big),
\end{equation}
where we have used the fact that $\ee^{-2\pi \ii \frac{t}{T} P} = \ee^{-2\pi \ii \frac{t}{T} }P + I -P$. Then we notice that, since $P^2 = P$
\begin{align}
P \Big[  \big[\Lambda_1, P \big],  \big[\Lambda_2, P \big] \Big] P 
&= - P \Lambda_1 (I-P) \Lambda_2 P + P \Lambda_2 (I-P) \Lambda_1 P \nonumber\\
& = P \Big[   \big[\Lambda_1, P \big], P \big[\Lambda_2, P \big] \Big] P + P \Big[P   \big[\Lambda_1, P \big],  \big[\Lambda_2, P \big] \Big] P
\end{align}
and
\begin{equation}
\Tr\Big( P \Big[  P \big[\Lambda_1, P \big], P \big[\Lambda_2, P \big] \Big] P \Big) = 0
\end{equation}
Then expanding $\ee^{-2\pi \ii \frac{t}{T} P} = \ee^{-2\pi \ii \frac{t}{T} }P + I -P$ in the trace of the previous integral we are left after some algebra with
\begin{align}
&\Tr\Big( P \Big[ \ee^{2\pi \ii \frac{t}{T} P} \big[\Lambda_1, \ee^{-2\pi \ii \frac{t}{T} P} \big], \ee^{2\pi \ii \frac{t}{T} P} \big[\Lambda_2, \ee^{-2\pi \ii \frac{t}{T} P} \big] \Big]\Big) \cr
& = 2 \Big(\cos \Big( \dfrac{2\pi t}{T}\Big) -1 \Big) \Tr\Big( P \Big[  \big[\Lambda_1, P \big],  \big[\Lambda_2, P \big] \Big] P \Big), 
\end{align}
which leads to \eqref{id_chern} after integration over $t$.

%%%%% BIBLIOGRAPHY %%%%%


\begin{thebibliography}{99}
	
\bibitem{AsbothTarasinskiDelplace14}
Asbóth, J. K., Tarasinski, B., and Delplace, P.:  {Chiral symmetry and bulk-boundary correspondence in periodically driven one-dimensional systems.} Physical Review B \textbf{90} (12), 125143 (2014).
	
\bibitem{AvronSeilerSimon94}
Avron, J., Seiler, R., and Simon, B.: {The index of a pair of projections.} Journal of Functional Analysis \textbf{120} (1), 220--237  (1994).

\bibitem{Lyon15bis}
Carpentier, D., Delplace, P., Fruchart, M., and Gaw\k{e}dzki, K.: {Topological index for periodically driven time-reversal invariant 2D systems.} Physical Review Letters \textbf{114} (10), 106806  (2015).

\bibitem{Lyon15}
Carpentier, D., Delplace, P., Fruchart, M., Gaw\k{e}dzki, K., and Tauber, C.:
\textit{Construction and properties of a topological index for periodically driven 
	time-reversal invariant 2D crystals.} Nuclear Physics B {\bf 896}, 779--834 (2015). 

\bibitem{ElgartGrafSchenker05}
Elgart, A., Graf, G. M., and Schenker, J. H.: {Equality of the bulk and edge Hall conductances in a mobility gap.} Communications in Mathematical Physics \textbf{259} (1),  185--221  (2005).

\bibitem{Fruchart16}
Fruchart, M.: { Complex classes of periodically driven topological lattice systems.} Physical Review B \textbf{93} (11), 115429  (2016).

\bibitem{FulgaMaksymenko16} 
Fulga, I. C., and Maksymenko, M.: {Scattering matrix invariants of Floquet topological insulators.} Physical Review B  \textbf{93} (7), 075405 (2016).

\bibitem{GrafPorta13}
Graf, G. M., Porta, M.: {Bulk-Edge Correspondence for Two-Dimensional Topological Insulators.} Communications in Mathematical Physics \textbf{324} (3),  851--895  (2013).

\bibitem{Hatsugai93}
Hatsugai, Y.: {Chern number and edge states in the integer quantum Hall effect.} Physical Review Letters \textbf{71} (22),  3697 (1993). 

\bibitem{InoueTanaka10}
Inoue, J. I., and  Tanaka, A.: {Photoinduced transition between conventional and topological insulators in two-dimensional electronic systems.} Physical Review Letters \textbf{105} (1),  017401 (2010).

\bibitem{KitgawaBergRudnerDemler10}
Kitagawa, T., Berg, E., Rudner, M., and Demler, E.: {Topological characterization of periodically driven quantum systems.} Physical Review B \textbf{82} (23),  235114 (2010). 

\bibitem{KlinovajaStanoLoss16}
Klinovaja, J., Stano, P., and Loss, D.: {Topological Floquet Phases in Driven Coupled Rashba Nanowires.} Physical Review Letters \textbf{116} (17), 176401 (2016).

\bibitem{LindnerRefaelGalitski11}	
Lindner, N. H., Refael, G., and Galitski, V.: {Floquet topological insulator in semiconductor quantum wells.} Nature Physics \textbf{7} (6),  490--495 (2011). 

\bibitem{Nathan16}
Nathan, F., Rudner, M. S., Lindner, N. H., Berg, E., Refael, G.: {Quantized magnetization density in periodically driven systems.} Physical review letters \textbf{119} (18), 186801.  (2016)

\bibitem{OkaAoki09}
Oka, T., and Aoki, H.:  {Photovoltaic Hall effect in graphene.} Physical Review B \textbf{79} (8), 081406 (2009). 

\bibitem{ProdanSchulz16book}
Prodan, E., and Schulz-Baldes, H.: {\it Bulk and Boundary Invariants for Complex Topological Insulators}. {Mathematical Physics Studies}, Springer,  2016.

\bibitem{ProdanSchulz16}
Prodan, E., and Schulz-Baldes, H.: {Non-commutative odd Chern numbers and topological phases of disordered chiral systems.}  Journal of Functional Analysis  \textbf{271} (5), 1150--1176 (2016).

\bibitem{ReedSimonII}
Reed, M., and Simon, B.: {\it Method of modern mathematical physics}. {Vol. II} Academic Press,  1980.

\bibitem{RudnerPRX13}
Rudner, M. S., Lindner, N. H., Berg, E., and Levin, M.:  {Anomalous edge states and the bulk-edge correspondence for periodically driven two-dimensional systems.}  Physical Review X {\bf 3} (3), 031005 (2013).

\bibitem{SadelSchulz17}
Sadel, C., and Schulz-Baldes, H.:  {Topological boundary invariants for Floquet systems and quantum walks.} Mathematical Physics, Analysis and Geometry \textbf{20} (4), 22 (2017).

\bibitem{TitumPRX16}
Titum, P., Berg, E., Rudner, M. S., Refael, G., and Lindner, N. H.: {Anomalous Floquet-Anderson insulator as a nonadiabatic quantized charge pump.} Physical Review X {\bf 6} (2), 021013 (2016).

\bibitem{Thoules83}
Thouless, D. J.: {Quantization of particle transport.} Physical Review B \textbf{27} (10),
 6083  (1983).



\end{thebibliography}
\end{document}